\documentclass[twocolumn,10pt]{article}
\usepackage[margin=0.8in]{geometry}
\usepackage{amsmath,amssymb,amsthm}
\usepackage{graphicx}
\usepackage{braket}
\usepackage[colorlinks=true,linkcolor=blue,citecolor=blue,urlcolor=blue]{hyperref}
\usepackage{booktabs}

\newtheorem{theorem}{Theorem}
\newtheorem{lemma}[theorem]{Lemma}
\newtheorem{proposition}[theorem]{Proposition}
\newtheorem{corollary}[theorem]{Corollary}
\theoremstyle{definition}
\newtheorem{definition}[theorem]{Definition}
\newtheorem{remark}[theorem]{Remark}

\newcommand{\id}{\mathbb{I}}
\newcommand{\M}{\mathcal{M}}
\newcommand{\C}{\mathbb{C}}
\newcommand{\R}{\mathbb{R}}
\newcommand{\PP}{\mathbb{P}}
\newcommand{\Dloc}{\mathcal{L}}
\newcommand{\Err}{\mathcal{E}}
\newcommand{\BT}{2\mathrm{T}}
\newcommand{\BO}{2\mathrm{O}}
\newcommand{\BI}{2\mathrm{I}}

\begin{document}

\title{\bf Error-Transparent Code Switching Between Permutation-Invariant Codes with Transversal Clifford and Icosahedral Gates}
\author{Zachary P.~Bradshaw\footnotemark[1]\\ {\small QodeX Quantum, Chicago, IL}}
\date{\today}

\twocolumn[
\begin{@twocolumnfalse}
\maketitle
\begin{abstract}
Code switching moves an encoded qubit between two codes whose transversal gates together are universal, but the switching operation itself is usually not transversal and can spread a single fault into an uncorrectable error. We study switching between two permutation-invariant codes on the same qubits, one with transversal Clifford gates and one with transversal gates from the non-Clifford binary icosahedral group. Both are compatible with transversal $X$, $Z$, and $F = HS^\dagger$, and the encoders with this property form a small linear space. We show that its codes of distance at least three are cut out by at most two quadratic equations, and that all of them have the same error geometry in the sense that their single-qubit error states have the same inner products. It follows that a switching unitary commuting with qubit permutations and with the shared transversal gates at every instant, and tolerating one single-qubit fault at an unknown time, exists if and only if the two codes are joined by a continuous path of distance-three codes. Such a switch can be made error transparent, so that a fault at any moment reaches the target code as the same single-qubit error. For real encoders, the distance-three codes are the solutions of a single quadratic equation, and on thirteen qubits, they form an explicit circle through two Clifford codes and an icosahedral code. We construct the switching unitary along this circle as a collective operator, leaving its synthesis into physical control open, and verify that its recovery infidelity per fault is at machine precision, whereas two natural alternatives leave average infidelities of $0.13$ and $0.85$. A survey over odd qubit numbers up to $29$ lists where such switches exist.
\end{abstract}
\vspace{1em}
\end{@twocolumnfalse}
]
{\renewcommand{\thefootnote}{\fnsymbol{footnote}}
\footnotetext[1]{zak@qodexquantum.ai}

\section{Introduction}
\label{sec:intro}

A logical gate is transversal if it is implemented by applying single-qubit gates to physical qubits separately. Such a gate cannot spread a fault on one qubit to other qubits, which makes it naturally fault tolerant. The Eastin--Knill theorem shows that no code admits a universal set of transversal gates~\cite{EastinKnill2009,ZengCrossChuang2011}, and for topological stabilizer codes~\cite{bradshaw2026introduction}, the Bravyi--K\"onig theorem restricts the available gates even further~\cite{BravyiKonig2013}. The most common way around this obstruction supplements transversal Clifford gates with magic states, prepared by distillation~\cite{BravyiKitaev2005} or, more recently, cultivation~\cite{Gidney2024}. A second way is code switching, in which the encoded qubit is moved between two codes whose transversal gates together generate a universal set~\cite{PaetznickReichardt2013,Anderson2014,Bombin2015}. Code switching has been developed into practical protocols~\cite{Butt2024,Heussen2025} and demonstrated experimentally~\cite{Pogorelov2024}, but for stabilizer codes, the cost of the switching step can offset much of the benefit of transversality~\cite{Beverland2021}.

Permutation-invariant (PI) codes are codes whose states do not change when the physical qubits are permuted~\cite{Ruskai2000,PollatsekRuskai2004,Ouyang2014,OuyangFitzsimons2016}. They are not stabilizer codes in general, and their transversal gates are therefore not limited to the Clifford hierarchy. However, PI codes and stabilizer codes can be united under the umbrella of symmetry-based quantum codes~\cite{bradshaw2026symmetry}, which have been used to construct an active error correction routine for simulations of SU(2) lattice gauge theories~\cite{bradshaw2026approximate}. Kubischta and Teixeira constructed PI codes on which applying one fixed gate to every qubit realizes the binary icosahedral group, a finite group of single-qubit gates that contains non-Clifford elements, as well as other groups that stabilizer codes cannot realize transversally~\cite{Kubischta2023,Kubischta2024,Kubischta2025}. They also found a thirteen-qubit PI code of distance three on which the transversal gates realize the full single-qubit Clifford group~\cite{Kubischta2025Intrinsic}. Ouyang, Jing, and Brennen proposed a measurement-free switching protocol between a stabilizer code and PI codes with transversal non-Clifford rotations, using logical teleportation through a bosonic mode~\cite{Ouyang2025Switch}. They showed that its gates are fault tolerant against single spin errors and against phase errors of the mode, while loss and gain errors of the mode propagate to the spins, so that a fully fault-tolerant implementation remains open. The switch studied here acts on a single block, with no ancilla code and no mediating mode, and its fault tolerance is of a different kind. Related structural descriptions of PI, spin, and bosonic codes appear in Refs.~\cite{Gross2021,AydinAlbertBarg2026,Kubischta2025Intrinsic,OuyangBrennen2026}.

Our starting point is a simple observation. The Clifford group and the binary icosahedral group, viewed as groups of single-qubit gates, share a subgroup, the binary tetrahedral group $\BT$ generated by $X$, $Z$, and $F = HS^\dagger$. A PI code whose transversal gates include the Clifford group and a PI code whose transversal gates include the icosahedral group therefore both have transversal $X$, $Z$, and $F$ acting as the logical $X$, $Z$, and $F$. We show that the encoders with this property form a small linear space $\M$, three-dimensional on thirteen qubits, and that both codes are points of it. A unitary that commutes with $X$, $Z$, and $F$ applied to every qubit maps $\M$ to itself without changing the logical state, so switching between the two codes amounts to moving a point through $\M$. The question is which motions are fault tolerant. Most paths through $\M$ pass through encoders that cannot correct single-qubit errors, and a fault that occurs while the encoder is at such a point becomes an uncorrectable error on the target code.

We answer this question for single-qubit faults. First we show that the distance-three conditions of Knill and Laflamme~\cite{KnillLaflamme1997}, applied to encoders in $\M$, reduce to at most two quadratic equations, and we compute them directly from the action of the transversal gates. Next we prove that all codes satisfying these equations have the same error geometry, meaning that the inner products between their error states are the same fixed numbers. This implies that the error states of any two such codes are related by a unitary that commutes with qubit permutations and with the transversal gates, and our main result follows. A switching unitary of this kind that tolerates one single-qubit fault at an unknown time exists if and only if the two codes are joined by a continuous path of distance-three codes, and in that case, the switch can be chosen error transparent, meaning that a fault at any time reaches the target code as the identical single-qubit error. Fault-tolerant switching is possible because the correctable-error geometry remains constant along the path of codes. The notion parallels error-transparent gates in bosonic and superconducting systems~\cite{Vy2013,Kapit2018,Ma2020}, and the construction uses a subspace transport going back to Kato~\cite{Kato1950,KatoBook}.

To apply the criterion, we restrict to encoders with real coefficients. For these, one of the two quadratic equations holds automatically, so the real distance-three codes are the zero set of a single quadratic form, and we show that this set is connected under mild conditions. On thirteen qubits, it is the circle $1485\alpha^2 - 330\alpha\beta + 5\beta^2 + 264\gamma^2 = 0$ in explicit coordinates, which passes through two Clifford codes and through a code with transversal icosahedral gates. We construct the error-transparent switch along this circle and compare it with two natural alternatives.

The paper is organized as follows. Section~\ref{sec:setting} sets up PI codes, the gate groups, and the space of encoders. Section~\ref{sec:quadrics} reduces the distance-three conditions to quadratic equations, and Section~\ref{sec:gram} proves that the error geometry is the same for all distance-three codes. Section~\ref{sec:switch} defines switches and proves the existence and necessity results. Section~\ref{sec:real} treats real encoders, proves the connectivity theorem, and reports a survey over qubit number. Section~\ref{sec:thirteen} works out thirteen qubits explicitly, Section~\ref{sec:numerics} presents the numerical construction and fault analysis, and Section~\ref{sec:discussion} discusses control errors, implementation, and extensions.

\section{Setting}
\label{sec:setting}

\paragraph{Permutation-invariant states.} Let $n$ be odd, let $j = n/2$, and let $\mathcal{H}_n = (\C^2)^{\otimes n}$ be the space of $n$ qubits. A state is permutation invariant if it does not change when the qubits are permuted. These states form the symmetric subspace $V_n$, which has dimension $n + 1$ and orthonormal basis given by the Dicke states
\begin{equation}
  \ket{D_w} = \binom{n}{w}^{-1/2} \sum_{|x| = w} \ket{x},
  \label{eq:dicke}
\end{equation}
for $w = 0, 1, \ldots, n,$ and where the sum runs over all bit strings with exactly $w$ ones. For a single-qubit unitary $g$, we write $D(g) = g^{\otimes n}$ for the gate $g$ applied to every qubit, which we call the transversal gate of $g$. It commutes with permutations of the qubits, so it maps $V_n$ to itself. We note that $D(X)\ket{D_w} = \ket{D_{n-w}}$ and $D(Z)\ket{D_w} = (-1)^w \ket{D_w}$.

We write $\sigma^a_k$ for the Pauli operator $\sigma^a \in \{\sigma^x=X, \sigma^y=Y, \sigma^z=Z\}$ acting on qubit $k$, and
\begin{equation}
  J_a = \frac{1}{2} \sum_{k=1}^{n} \sigma^a_k, \qquad a \in \{x, y, z\},
\end{equation}
for the collective spin operators. On $V_n$, these operators satisfy $J_x^2 + J_y^2 + J_z^2 = j(j+1)$, and $J_z \ket{D_w} = (j - w) \ket{D_w}$. Conjugating by a transversal gate acts on each qubit separately, so the single-qubit rules $XYX = -Y$, $XZX = -Z$, $ZXZ = -X$, $ZYZ = -Y$, and
\begin{equation}
  F X F^\dagger = Y, \qquad F Y F^\dagger = Z, \qquad F Z F^\dagger = X,
  \label{eq:Fcycle}
\end{equation}
with $F = HS^\dagger$, give
\begin{equation}
\begin{aligned}
  D(X)\,(J_x, J_y, J_z)\,D(X)^\dagger &= (J_x, -J_y, -J_z), \\
  D(Z)\,(J_x, J_y, J_z)\,D(Z)^\dagger &= (-J_x, -J_y, J_z), \\
  D(F)\,(J_x, J_y, J_z)\,D(F)^\dagger &= (J_y, J_z, J_x),
\end{aligned}
\label{eq:Jrules}
\end{equation}
where each line is shorthand for three equations, for example $D(X) J_y D(X)^\dagger = -J_y$. In general, conjugation by a single-qubit unitary $g$ rotates the Pauli vector, $g \sigma^a g^\dagger = \sum_b R(g)_{ba} \sigma^b$ for a rotation matrix $R(g)$, and so it follows that $D(g) J_a D(g)^\dagger = \sum_b R(g)_{ba} J_b$.

\paragraph{Gate groups.} Let $H$ be the Hadamard gate, $S = \mathrm{diag}(1, i)$, $F = H S^\dagger$, and
\begin{equation}
  \Phi = \frac{1}{2}\begin{pmatrix} \varphi + i\varphi^{-1} & 1 \\ -1 & \varphi - i \varphi^{-1} \end{pmatrix},
\end{equation}
where $\varphi = \frac{1+\sqrt{5}}{2}$ is the golden ratio. The gate $\Phi$ rotates the Bloch sphere by $72^\circ$ and is not a Clifford gate. To obtain finite groups, we fix the overall phase of each gate by rescaling it to determinant one, writing $\hat g = g/\sqrt{\det g}$ with the principal square root. This does not change how a gate acts on states, only its global phase. Explicitly, $\hat X = -iX$, $\hat Z = -iZ$, $\hat F = e^{-i\pi/4} F$, and $\hat S = e^{-i\pi/4} S$, while $\Phi$ already has determinant one. The binary tetrahedral group $\BT$ is the group generated by $\hat X$, $\hat Z$, and $\hat F$, and it has $24$ elements. Adding $\hat S$ gives the binary octahedral group $\BO$ with $48$ elements, which up to phases is the single-qubit Clifford group. Adding $\Phi$ to $\BT$ instead gives the binary icosahedral group $\BI$ with $120$ elements~\cite{ConwaySmith}. Up to sign, these are the rotations of a tetrahedron, a cube, and an icosahedron, and $\BO \cap \BI = \BT$.

\paragraph{Encoders.} Let $W = \C^2$ be the logical qubit, with basis $e_0, e_1$. An encoder is a linear map $A \colon W \to V_n$ that preserves inner products. Its image $\mathcal{C}_A$ is the code, and $A e_0$, $A e_1$ are the encoded basis states $\ket{0_L}$, $\ket{1_L}$. The transversal gate $D(g)$ implements the logical gate $L$ if $D(g) A = A L$, meaning that applying $g$ to every qubit of an encoded state $A\psi$ gives the encoded state $AL\psi$.

\begin{definition}
Let $\M(n)$ be the set of linear maps $A \colon W \to V_n$ with
\begin{equation}
  D(\hat X) A = A \hat X, \quad D(\hat Z) A = A \hat Z, \quad D(\hat F) A = A \hat F.
  \label{eq:Mdef}
\end{equation}
\end{definition}

The set $\M(n)$ is a vector space, since sums and multiples of such maps satisfy the same equations. If $A$ satisfies~\eqref{eq:Mdef}, then it satisfies $D(k) A = A k$ for every $k \in \BT$ because these elements are products of the three generators. The dimension $m:=\dim(\M(n))$ is listed for $n \le 29$ in Table~\ref{tab:survey}.

Every nonzero $A \in \M(n)$ is a multiple of an encoder. Indeed, $A^\dagger A$ commutes with $\hat X$ and $\hat Z$, because $A^\dagger A \hat X = A^\dagger D(\hat X) A = \hat X A^\dagger A$ and similarly for $\hat Z$, and the only $2 \times 2$ matrices commuting with $X$ and $Z$ are multiples of the identity. Hence $A^\dagger A = \lVert A e_0 \rVert^2 \id$, and $A/\lVert A e_0 \rVert$ is an encoder. Two maps that differ by a nonzero scalar give the same code, so a code is an element of $\M(n)$ up to a scalar. We write $[A]$ for the code determined by $A$ and $\PP(\M(n))$ for the projective space of all such $[A]$. For $A, B \in \M(n)$, we define the inner product $\braket{A, B} = \braket{A e_0 | B e_0}$. The same computation with $A^\dagger B$ in place of $A^\dagger A$ shows that $A^\dagger B$ commutes with $X$ and $Z$, so it is a multiple of $\id$, and its top-left entry is $\braket{A, B}$. Hence $A^\dagger B = \braket{A, B}\id$, which says that $\braket{A e_1 | B e_1} = \braket{A e_0 | B e_0}$ and $\braket{A e_0 | B e_1} = 0$.

\paragraph{Clifford and icosahedral codes.} To say which logical gate each transversal gate performs, we use representations. A 2-dimensional representation of a group $G$ of gates is a rule $\rho$ that assigns to each $g \in G$ a $2 \times 2$ unitary $\rho(g)$, with $\rho(gh) = \rho(g)\rho(h)$ for all $g, h \in G$. An encoder $A$ is compatible with $\rho$ if $D(g) A = A \rho(g)$ for all $g \in G$, so that the transversal gate $D(g)$ performs the logical gate $\rho(g)$. The simplest example is the defining representation $\rho(g) = g$, in which each transversal gate performs the same gate on the logical qubit. The encoders in $\M(n)$ are exactly those compatible with the defining representation of $\BT$.

For the Clifford group $\BO$, two representations agree with the defining one on $\BT$. The first is the defining representation itself, which we denote $\rho_{1/2}$. The second, denoted $\rho_{5/2}$, assigns $g$ to the elements of $\BT$ and $-g$ to the other elements of $\BO$. It is again a representation because the product of two elements of $\BO$ lies in $\BT$ exactly when both or neither of them do, so the signs multiply correctly. The subscripts are the standard names and record that these representations arise from spin $\frac12$ and spin $\frac52$; for these labels and the Galois labels below see Refs.~\cite{Kubischta2023,Kubischta2024}. A Clifford code is an encoder in $\M(n)$ compatible with $\rho_{1/2}$ or with $\rho_{5/2}$. Since $\BO$ is generated by $\BT$ and $\hat S$, this means $D(\hat S) A = A \hat S$ or $D(\hat S) A = -A \hat S$. In either case, the logical gates are the Clifford gates, up to a sign, which is a global phase.

For the icosahedral group $\BI$, the defining representation, again denoted $\rho_{1/2}$, has $D(\Phi)$ performing $\Phi$. A second representation $\rho'$ is obtained by replacing $\sqrt5$ with $-\sqrt5$ in every matrix entry. The entries of the elements of $\BI$ are of the form $u + v\sqrt5$ with $u$ and $v$ complex rational numbers, and this replacement respects sums and products of such numbers, so $\rho'$ satisfies $\rho'(gh) = \rho'(g)\rho'(h)$. It leaves the elements of $\BT$ unchanged, since their entries do not involve $\sqrt5$, and since $\varphi \mapsto -\varphi^{-1}$ under the replacement, it sends $\Phi$ to
\begin{equation}
  \Phi' = \frac12 \begin{pmatrix} -\varphi^{-1} - i\varphi & 1 \\ -1 & -\varphi^{-1} + i\varphi \end{pmatrix}.
\end{equation}
Replacements of this kind are called Galois conjugations, and $\rho'$ is the Galois-conjugate representation. The gate $\Phi'$ is a rotation by $144^\circ$ and is not a Clifford gate, and together with $\BT$ it generates a different group of $120$ gates isomorphic to $\BI$ and equally non-Clifford. An icosahedral code is an encoder in $\M(n)$ compatible with $\rho_{1/2}$ or with $\rho'$, that is, with $D(\Phi) A = A \Phi$ or $D(\Phi) A = A \Phi'$.

We write $\M_\rho(n) \subset \M(n)$ for the encoders compatible with each of these four representations $\rho$. The Clifford and icosahedral codes on $n$ qubits are therefore all points of the single projective space $\PP(\M(n))$.

\section{Distance three as quadratic equations}
\label{sec:quadrics}

We first recall the Knill--Laflamme conditions for PI codes. The following lemma is standard~\cite{PollatsekRuskai2004,Ouyang2014,Gross2021}, but we include the proof because the identities in it are used again below.

\begin{lemma}
\label{lem:PIKL}
Let $A\colon W \to V_n$ be an encoder. The code $\mathcal{C}_A$ corrects arbitrary single-qubit errors if and only if there are numbers $c_a$ and $c_{ab}$, for $a, b \in \{x, y, z\}$, with
\begin{equation}
  A^\dagger J_a A = c_a \id, \qquad A^\dagger \{J_a, J_b\} A = c_{ab} \id.
  \label{eq:PIKL}
\end{equation}
\end{lemma}

\begin{proof}
Let $P = AA^\dagger$ be the orthogonal projection onto the code, and note that $A^\dagger A = \id$. The Knill--Laflamme condition~\cite{KnillLaflamme1997} requires $P E^\dagger E' P = c\,P$ for all errors $E, E'$ in the set, with a number $c$ depending on $E$ and $E'$. Multiplying by $A^\dagger$ on the left and by $A$ on the right turns this into $A^\dagger E^\dagger E' A = c\,\id$, and conversely $A^\dagger E^\dagger E' A = c\,\id$ gives back $P E^\dagger E' P = A(A^\dagger E^\dagger E' A)A^\dagger = cP$. Since the $(i,i')$ entry of $A^\dagger O A$ is $\braket{i_L | O | i'_L}$, this is the familiar form $\braket{i_L | E^\dagger E' | i'_L} = c\,\delta_{ii'}$, and we work with it throughout.

For single-qubit errors, the set is $\Err_1 = \{\id\} \cup \{\sigma^a_k\}$. The Pauli operators are Hermitian, so the products $E^\dagger E'$ are the identity, the operators $\sigma^a_k$, the products $\sigma^a_k \sigma^b_k$ on one qubit, and the products $\sigma^a_k \sigma^b_l$ on two different qubits. On a single qubit $\sigma^a \sigma^b = \delta_{ab}\id + i\sum_c \epsilon_{abc}\sigma^c$, so the products on one qubit reduce to the first two kinds. The code therefore corrects all single-qubit errors exactly when $A^\dagger \sigma^a_k A$ and, for $k \ne l$, $A^\dagger \sigma^a_k \sigma^b_l A$ are multiples of $\id$.

Permutation invariance reduces these conditions to statements about collective operators. If $\pi$ is the unitary that permutes the qubits, then $\pi A = A$, and hence $A^\dagger O A = A^\dagger \pi^\dagger O \pi A$ for every operator $O$. Conjugation by $\pi$ only relabels qubits, $\pi^\dagger \sigma^a_k \pi = \sigma^a_{\pi^{-1}(k)}$, so $A^\dagger \sigma^a_k A$ does not depend on $k$, and $A^\dagger \sigma^a_k \sigma^b_l A$ does not depend on the ordered pair $k \ne l$. Each is therefore equal to its average, which gives
\begin{equation}
  A^\dagger \sigma^a_k A = \frac{1}{n}\sum_{k'} A^\dagger \sigma^a_{k'} A = \frac{2}{n} A^\dagger J_a A
\end{equation}
and $A^\dagger \sigma^a_k \sigma^b_l A = A^\dagger T_{ab} A/(n(n-1))$ with $T_{ab} = \sum_{k \ne l} \sigma^a_k \sigma^b_l$.

It remains to express $T_{ab}$ through the collective operators. Writing $4J_aJ_b$ as a double sum over qubits and separating the terms with equal indices,
\begin{equation}
\begin{aligned}
  4 J_a J_b &= \sum_{k,l} \sigma^a_k \sigma^b_l = T_{ab} + \sum_k \sigma^a_k \sigma^b_k \\
  &= T_{ab} + n\delta_{ab}\id + 2i\sum_c \epsilon_{abc} J_c,
\end{aligned}
\end{equation}
where the last equality uses the one-qubit product rule. Hence
\begin{equation}
  T_{ab} = 4 J_a J_b - n\delta_{ab}\id - 2i\sum_c \epsilon_{abc} J_c .
  \label{eq:Tab}
\end{equation}
The first family of conditions now says that $A^\dagger J_a A$ is a multiple of $\id$. When it holds, the last term in~\eqref{eq:Tab} contributes a multiple of $\id$, so the second family says that $A^\dagger J_a J_b A$ is a multiple of $\id$. Since $J_a J_b = \frac12\{J_a, J_b\} + \frac{i}{2}\sum_c \epsilon_{abc} J_c$, and the commutator term again contributes a multiple of $\id$, this is equivalent to $A^\dagger \{J_a, J_b\} A$ being a multiple of $\id$. The conditions are therefore exactly~\eqref{eq:PIKL}.
\end{proof}

For an encoder $A \in \M(n)$, the $2 \times 2$ matrices in~\eqref{eq:PIKL} are far from arbitrary. Define the two real numbers
\begin{equation}
  c_1(A) = \braket{0_L | J_z | 0_L}, \qquad c_2(A) = \braket{0_L | \{J_x, J_y\} | 0_L},
  \label{eq:c1c2}
\end{equation}
where $\ket{0_L} = A e_0$. Both are quadratic in $A$, and they are real because $J_z$ and $\{J_x, J_y\}$ are Hermitian.

\begin{proposition}
\label{prop:reduction}
Let $A \in \M(n)$ be an encoder. Then
\begin{equation}
  A^\dagger J_a A = c_1(A)\, \sigma^a, \qquad A^\dagger \{J_b, J_c\} A = c_2(A)\, \sigma^a
  \label{eq:equiv}
\end{equation}
for every cyclic permutation $(a, b, c)$ of $(x, y, z)$, and $A^\dagger J_a^2 A = \frac{j(j+1)}{3} \id$ for every $a$. Consequently the code $\mathcal{C}_A$ has distance at least three if and only if $c_1(A) = c_2(A) = 0$, and the distance-three codes form the set
\begin{equation}
  \Dloc(n) = \{[A] \in \PP(\M(n)) \mid c_1(A) = c_2(A) = 0\}.
\end{equation}
\end{proposition}

\begin{proof}
All three statements follow from one computation. If $O$ is an operator on $V_n$ and $g$ is one of the generators $\hat X, \hat Z, \hat F$, then~\eqref{eq:Mdef} gives
\begin{equation}
  A^\dagger \bigl(D(g)^\dagger O D(g)\bigr) A = g^\dagger (A^\dagger O A)\, g ,
  \label{eq:conjrule}
\end{equation}
and more generally $A^\dagger (D(g)^\dagger O D(g)) A = \rho(g)^\dagger (A^\dagger O A) \rho(g)$ whenever $A$ is compatible with a representation $\rho$ and $g$ is in its group.
Let $\Lambda_a = A^\dagger J_a A$. By~\eqref{eq:Jrules}, conjugation by $D(X)$ fixes $J_x$ and negates $J_y, J_z$, and conjugation by $D(Z)$ negates $J_x, J_y$ and fixes $J_z$. Applying~\eqref{eq:conjrule} with $O = J_z$ shows that $\Lambda_z$ commutes with $Z$ and anticommutes with $X$, so $\Lambda_z$ is diagonal with opposite diagonal entries, that is, $\Lambda_z = c\,\sigma^z$ with $c = \braket{0_L | J_z | 0_L} = c_1(A)$. In the same way, $\Lambda_x$ is a multiple of $\sigma^x$ and $\Lambda_y$ a multiple of $\sigma^y$. The gate $\hat F$ permutes $(J_x, J_y, J_z)$ cyclically by~\eqref{eq:Jrules} and permutes $(\sigma^x, \sigma^y, \sigma^z)$ in the same way by~\eqref{eq:Fcycle}, so the three multiples agree. This proves the first relation.

For the second, let $\Lambda_{xy} = A^\dagger \{J_x, J_y\} A$. Conjugation by $D(X)$ negates $\{J_x, J_y\}$ and conjugation by $D(Z)$ fixes it, so $\Lambda_{xy}$ anticommutes with $X$ and commutes with $Z$, and hence $\Lambda_{xy} = c_2(A)\, \sigma^z$. Conjugation by $D(\hat F)$ then gives $A^\dagger \{J_y, J_z\} A = c_2(A)\sigma^x$ and $A^\dagger \{J_z, J_x\} A = c_2(A)\sigma^y$. Finally, $J_a^2$ is fixed by conjugation with $D(X)$ and $D(Z)$, so $A^\dagger J_a^2 A$ commutes with $X$ and $Z$ and is a multiple of $\id$. Conjugation by $D(\hat F)$ shows that the three multiples are equal, and since $J_x^2 + J_y^2 + J_z^2 = j(j+1)$ on $V_n$, each equals $\frac{j(j+1)}{3}$. Every condition in~\eqref{eq:PIKL} is now expressed through $c_1$ and $c_2$. The matrices $c_1 \sigma^a$ and $c_2 \sigma^a$ are multiples of $\id$ only when $c_1 = c_2 = 0$, and Lemma~\ref{lem:PIKL} completes the proof.
\end{proof}

We call $\Dloc(n)$ the set of distance-three codes, by which we mean codes of distance at least three, and we use the phrase in that sense throughout. For Clifford and icosahedral codes, the extra transversal gate removes one or both of the two conditions.

\begin{corollary}
\label{cor:larger}
Let $A \in \M_\rho(n)$ be an encoder.
\begin{enumerate}
  \item If $A$ is a Clifford code, or an icosahedral code with $D(\Phi)$ implementing $\Phi$, then $c_2(A) = 0$, and $\mathcal{C}_A$ has distance at least three if and only if $c_1(A) = 0$.
  \item If $A$ is an icosahedral code with $D(\Phi)$ implementing $\Phi'$, then $c_1(A) = c_2(A) = 0$, so $\mathcal{C}_A$ always has distance at least three.
\end{enumerate}
\end{corollary}

\begin{proof}
For a Clifford code, $D(\hat S) A = \pm A \hat S$ and the rule~\eqref{eq:conjrule} holds for $g = \hat S$ with $\rho(\hat S) = \pm \hat S$, and the sign cancels. From $S^\dagger X S = -Y$ and $S^\dagger Y S = X$, we get $D(S)^\dagger \{J_x, J_y\} D(S) = -\{J_x, J_y\}$, so by Proposition~\ref{prop:reduction}, $-c_2 \sigma^z = S^\dagger (c_2 \sigma^z) S = c_2 \sigma^z$ and $c_2 = 0$.

For an icosahedral code, we use the same kind of computation with the extra gate $\Phi$. For a vector $u \in \R^3$ write $u \cdot J = u_x J_x + u_y J_y + u_z J_z$. Expanding the square and pairing the terms with $a \ne b$, we have $(u \cdot J)^2= \sum_a u_a^2 J_a^2 + u_y u_z \{J_y, J_z\} + u_z u_x \{J_z, J_x\} + u_x u_y \{J_x, J_y\}$.
Sandwiching with $A^\dagger$ and $A$ and using Proposition~\ref{prop:reduction} for each term gives
\begin{equation}
  A^\dagger (u \cdot J)^2 A = \tfrac{j(j+1)}{3}\, \lvert u \rvert^2\, \id + c_2\, v(u) \cdot \sigma,
\end{equation}
with $v(u) = (u_y u_z,\ u_z u_x,\ u_x u_y)$. Conjugating $(u \cdot J)^2$ by $D(\Phi)$ replaces $u$ by $R^{-1}u$, where $R = R(\Phi)$, while the compatibility $D(\Phi) A = A \rho(\Phi)$ conjugates the right side by $\rho(\Phi)$, which replaces $v(u)$ by $R_\rho^{-1} v(u)$ with $R_\rho$ the rotation of the logical gate $\rho(\Phi)$. Hence $c_2\, v(R^{-1}u) = c_2\, R_\rho^{-1} v(u)$ for every $u$. Take $u = e_x$, the unit vector along $x$. Then $v(e_x) = 0$, while $R^{-1} e_x = \frac12(\varphi^{-1}, 1, -\varphi)$ has three nonzero components, so $v(R^{-1}e_x) = \frac14(-\varphi, -1, \varphi^{-1}) \ne 0$, and therefore $c_2 = 0$. This proves part 1 for icosahedral codes and the statement $c_2 = 0$ in part 2, for either choice of $\rho(\Phi)$.

For part 2, it remains to show $c_1 = 0$. The rotation rule $D(\Phi)^\dagger (u \cdot J) D(\Phi) = (R^{-1} u) \cdot J$, with $R = R(\Phi)$ and $u \in \R^3$, together with $D(\Phi) A = A \Phi'$ and the first relation in~\eqref{eq:equiv}, gives
\begin{equation}
  c_1\, (R^{-1} u) \cdot \sigma = c_1\, \Phi'^\dagger (u \cdot \sigma) \Phi' = c_1\, (R'^{-1} u) \cdot \sigma,
\end{equation}
where $R' = R(\Phi')$. The rotations $R$ and $R'$ are different, one by $72^\circ$ and the other by $144^\circ$, so $c_1 = 0$.
\end{proof}

Part 2 says that an icosahedral code implementing $\Phi'$ has distance three automatically, while one implementing $\Phi$ must satisfy $c_1 = 0$. Table~\ref{tab:survey} shows that on seven qubits, no encoder implementing $\Phi$ exists, and on thirteen qubits, the one that exists has $c_1 \ne 0$. This is why the known distance-three codes with transversal icosahedral gates, including the code of Pollatsek and Ruskai~\cite{PollatsekRuskai2004}, implement $\Phi'$; see also Ref.~\cite{Kubischta2023}. Corollary~\ref{cor:larger} is consistent with the description in Ref.~\cite{Kubischta2025Intrinsic}, which organizes error-protection conditions by symmetry type.

\section{All distance-three codes have the same error geometry}
\label{sec:gram}

A code corrects a set of errors when the states obtained by applying these errors to the codewords have inner products that do not depend on the logical state. The matrix of these inner products, the error Gram matrix, determines how the recovery operation works. The central fact of this paper is that this matrix is the same for every distance-three code in $\M(n)$. Let $\Err_1 = \{\id\} \cup \{\sigma^a_k\}$ denote the $1 + 3n$ single-qubit Pauli errors, including the identity, indexed by $\mu \in \{0\} \cup \{(a, k)\}$, and write $E_\mu$ for the error with index $\mu$.

\begin{theorem}
\label{thm:gram}
Let $A \in \M(n)$ be an encoder with $[A] \in \Dloc(n)$. For all $\mu, \nu$ and $i, i' \in \{0, 1\}$,
\begin{equation}
  \braket{E_\mu A e_i | E_\nu A e_{i'}} = \delta_{ii'}\, \Gamma_{\mu\nu},
\end{equation}
where $\Gamma_{00} = 1$, $\Gamma_{0,(a,k)} = 0$, and
\begin{equation}
  \Gamma_{(a,k),(b,l)} = \delta_{ab} \left( \frac{2}{3} \delta_{kl} + \frac{1}{3} \right).
  \label{eq:Gamma}
\end{equation}
The matrix $\Gamma$ is real, depends only on $n$, and in particular does not depend on $A$. It is invertible, with eigenvalues $1$, $\frac{2}{3}$, and $\frac{n+2}{3}$.
\end{theorem}

\begin{proof}
By Proposition~\ref{prop:reduction}, $A^\dagger J_a A = 0$, $A^\dagger \{J_a, J_b\} A = 0$ for $a \ne b$, and $A^\dagger J_a^2 A = \frac{j(j+1)}{3}\id$. The entries with $\mu = 0$ follow from $A^\dagger \sigma^a_k A = \frac{2}{n} A^\dagger J_a A = 0$. For $k = l$, we have $A^\dagger \sigma^a_k \sigma^b_k A = \delta_{ab}\id + i\sum_c\epsilon_{abc} A^\dagger \sigma^c_k A = \delta_{ab}\id$. For $k \ne l$, Eq.~\eqref{eq:Tab} gives
\begin{equation}
  A^\dagger \sigma^a_k \sigma^b_l A = \frac{4 A^\dagger J_a J_b A - n\delta_{ab}\id}{n(n-1)} = \delta_{ab} \frac{\frac{4}{3} j (j+1) - n}{n(n-1)} \id,
\end{equation}
and $\frac{4}{3} j(j+1) - n = \frac{n(n+2)}{3} - n = \frac{n(n-1)}{3}$, which gives the value $\frac13$. For fixed $a$, the $n \times n$ block of $\Gamma$ is $\frac23 \id_n + \frac13 \mathbf{u}\mathbf{u}^{\mathsf T}$ with $\mathbf{u} = (1, \ldots, 1)$, whose eigenvalues are $\frac23$ and $\frac{n+2}{3}$.
\end{proof}

The theorem says that every code in $\Dloc(n)$ sees single-qubit errors in exactly the same way. Take two codes $A$ and $A'$ in $\Dloc(n)$ and match their error states, $E_\mu A e_i \mapsto E_\mu A' e_i$. By the theorem, any two error states of $A$ have the same inner product as the matched pair of error states of $A'$. A correspondence that preserves all inner products extends to a unitary map, so there is a unitary carrying the error space of $A$ onto the error space of $A'$ and sending each error state to the matched one. The switch constructed in the next section is such a unitary, chosen to depend continuously on the code and to respect the symmetries.

The same inner products determine how errors are corrected, and since they are the same for every code in $\Dloc(n)$, one recovery procedure serves all of them. Because $\Gamma$ is invertible, the $1 + 3n$ error states of each logical basis state are linearly independent, and the vectors
\begin{equation}
  w_{\eta, i} = \sum_\mu (\Gamma^{-1/2})_{\mu\eta}\, E_\mu A e_i,
\end{equation}
with $\eta \in \{0\} \cup \{(a, k)\},\ i \in \{0, 1\}$ are orthonormal since $\braket{w_{\eta,i} | w_{\eta',i'}} = \delta_{ii'} (\Gamma^{-1/2} \Gamma\, \Gamma^{-1/2})_{\eta\eta'} = \delta_{ii'}\delta_{\eta\eta'}$. For fixed $\eta$, the two vectors $w_{\eta,0}, w_{\eta,1}$ span a copy of the logical qubit, which we call the syndrome subspace with syndrome $\eta$; the error space is the orthogonal sum of these $1 + 3n$ subspaces and has dimension $2(1 + 3n)$. The standard recovery of $\mathcal{C}_A$, used in Proposition~\ref{prop:ET} and Sec.~\ref{sec:numerics}, measures which syndrome subspace the state lies in, maps that subspace back onto the code by $w_{\eta,i} \mapsto A e_i$, and reports a state outside the error space as an erasure.

\begin{remark}
\label{rem:structural}
Theorem~\ref{thm:gram} has a short explanation that also covers higher distance and other groups. The numbers entering the Gram matrix are the values of $\frac12 \mathrm{tr}(A^\dagger O A)$ on polynomials $O$ of degree at most two in the $J_a$, and by~\eqref{eq:conjrule} these values do not change when $O$ is conjugated by a transversal gate of $\BT$. The only polynomial of degree at most two that is invariant under the tetrahedral rotations is $J_x^2 + J_y^2 + J_z^2$, up to scalars, and its value $j(j+1)$ is fixed, so the Gram matrix is fixed. At distance five, the polynomials have degree up to four, and the tetrahedral group has invariants of degree three and four, namely $K_3$ and $Q_4$ of Sec.~\ref{sec:numerics}, whose values on the code can vary. A dihedral group already has the invariant $J_z^2$ of degree two, so for its codes the Gram matrix can vary at distance three; see Sec.~\ref{sec:discussion}. For any finite group of gates, the error geometry at distance $2t+1$ is fixed on the set of codes whenever its rotation image has no invariant polynomial of degree $1$ through $2t$ in the $J_a$ other than powers of $J_x^2 + J_y^2 + J_z^2$.
\end{remark}

\section{Error-transparent switching}
\label{sec:switch}

\paragraph{Switches.} We allow switching operations that treat all qubits alike and that respect the shared transversal gates.

\begin{definition}
\label{def:switch}
A switch is a continuous family of unitaries $U(\lambda)$ on $\mathcal{H}_n$, for $\lambda \in [0,1]$, with $U(0) = \id$, such that every $U(\lambda)$ commutes with all permutations of the qubits and with $D(\hat X)$, $D(\hat Z)$, and $D(\hat F)$. Given an encoder $A_0 \in \M(n)$, its trajectory is $A_\lambda = U(\lambda) A_0$.
\end{definition}

The two requirements have simple consequences. Because $U(\lambda)$ commutes with permutations, it maps $V_n$ to itself, and it is a collective operation, since the operators commuting with all qubit permutations are exactly the polynomials in the collective spin operators $J_x, J_y, J_z$~\cite{GoodmanWallach}. Because $U(\lambda)$ commutes with the transversal gates, the trajectory stays in $\M(n)$ and the logical state is never touched. Indeed, $D(\hat X) A_\lambda = U(\lambda) D(\hat X) A_0 = U(\lambda) A_0 \hat X = A_\lambda \hat X$, and similarly for $\hat Z$ and $\hat F$. An encoded state $A_0 \psi$ is therefore carried to $A_\lambda \psi$ with the same logical state $\psi$ at every time. A switch from $\mathcal{C}_{A_0}$ to $\mathcal{C}_{A_1}$ is a switch whose trajectory ends at a multiple of $A_1$ by a phase, which is a global logical phase and has no physical effect.

\paragraph{Fault model.} A single fault is a Kraus operator $E$, a linear combination of the single-qubit Pauli errors in $\Err_1$, acting at one instant $\lambda$ during the switch. Noise acting continuously during the switch lies outside this model, except at first order in its strength, where it decomposes into such single faults. A fault at time $\lambda$ turns the final encoded state into $U_{\lambda \to 1} E A_\lambda \psi$, where $U_{\lambda \to 1} = U(1) U(\lambda)^\dagger$ is the part of the switch that remains after the fault. From the point of view of the target code, the fault therefore appears as the error $U_{\lambda \to 1} E U_{\lambda \to 1}^\dagger$ acting on $A_1 \psi$.

\begin{definition}
A switch is \textbf{pointwise single-fault tolerant} if, for every $\lambda \in [0,1]$, the target code $\mathcal{C}_{A_1}$ satisfies the Knill--Laflamme conditions for the error set $U_{\lambda \to 1} \Err_1 U_{\lambda \to 1}^\dagger$. It is \textbf{single-fault tolerant} if $\mathcal{C}_{A_1}$ satisfies the Knill--Laflamme conditions for the full error set
\begin{equation}
  \Err_{\mathrm{sw}} = \{ U_{\lambda \to 1} E U_{\lambda \to 1}^\dagger \mid E \in \Err_1,\ \lambda \in [0,1] \},
  \label{eq:Esw}
\end{equation}
including the conditions that mix different fault times,
\begin{equation}
  A_1^\dagger (U_{\lambda \to 1} E U_{\lambda \to 1}^\dagger)^\dagger\, U_{\lambda' \to 1} E' U_{\lambda' \to 1}^\dagger A_1 \in \C \id,
\end{equation}
with $\lambda \ne \lambda'$. It is \textbf{error transparent} if
\begin{equation}
  U(\lambda)\, E\, A_0 = E\, A_\lambda
  \label{eq:ET}
\end{equation}
for all $E \in \Err_1,\ \lambda \in [0,1]$.
\end{definition}

Pointwise tolerance allows the recovery to depend on the time of the fault. Single-fault tolerance is the operational requirement, since in a computation the time of a fault is unknown and one recovery must work for all times. Error transparency says that the switch carries each error state of the initial code to the corresponding error state of the current code, so that a fault can be pushed through the rest of the switch without changing its form. It implies single-fault tolerance.

\begin{proposition}
\label{prop:ET}
Let the switch be error transparent. A fault $E \in \Err_1$ at any time $\lambda$ produces the final encoded map $E A_1$, the same as the fault $E$ occurring after the switch. If $[A_0] \in \Dloc(n)$, then $[A_\lambda] \in \Dloc(n)$ for all $\lambda$, the switch is single-fault tolerant, and the standard recovery of $\mathcal{C}_{A_1}$ corrects one single-qubit fault at an unknown time.
\end{proposition}

\begin{proof}
Applying~\eqref{eq:ET} at $\lambda$ and at $1$ gives $U_{\lambda \to 1} E A_\lambda = U(1) U(\lambda)^\dagger U(\lambda) E A_0 = U(1) E A_0 = E A_1$. Since $U(\lambda)$ is unitary, the inner products among the states $E A_\lambda e_i$ equal those among the states $E A_0 e_i$, which are $\delta_{ii'} \Gamma_{\mu\nu}$ by Theorem~\ref{thm:gram}, so $\mathcal{C}_{A_\lambda}$ satisfies the Knill--Laflamme conditions. For the full error set, $A_1 = U_{\lambda \to 1} A_\lambda$ gives $U_{\lambda \to 1} E U_{\lambda \to 1}^\dagger A_1 = U_{\lambda \to 1} E A_\lambda = E A_1$ for every $\lambda$. Every element of $\Err_{\mathrm{sw}}$ therefore acts on the target code as an element of $\Err_1$, the conditions mixing different times reduce to $A_1^\dagger E^\dagger E' A_1 \in \C\id$, and these hold because $[A_1] \in \Dloc(n)$.
\end{proof}

The converse direction starts from a simple necessary condition, which holds for every switch.

\begin{proposition}
\label{prop:necessary}
If a switch is pointwise single-fault tolerant, and in particular if it is single-fault tolerant, then $[A_\lambda] \in \Dloc(n)$ for all $\lambda \in [0,1]$.
\end{proposition}

\begin{proof}
The Knill--Laflamme conditions for $\mathcal{C}_{A_1}$ with respect to $U_{\lambda \to 1} \Err_1 U_{\lambda \to 1}^\dagger$ read $A_1^\dagger U_{\lambda \to 1} E^\dagger E' U_{\lambda \to 1}^\dagger A_1 \in \C\id$ for $E, E' \in \Err_1$. Substituting $A_1 = U_{\lambda \to 1} A_\lambda$ turns these into $A_\lambda^\dagger E^\dagger E' A_\lambda \in \C \id$, which say that $\mathcal{C}_{A_\lambda}$ corrects single-qubit errors. Proposition~\ref{prop:reduction} then gives $[A_\lambda] \in \Dloc(n)$.
\end{proof}

Proposition~\ref{prop:necessary} says that fault tolerance of the switch is decided by the codes it passes through. Suppose that at some time $\lambda$ the code $\mathcal{C}_{A_\lambda}$ fails to correct some single-qubit error. A fault at that time is carried by the rest of the switch to the target code, and by the proof, the target code corrects the carried error exactly when $\mathcal{C}_{A_\lambda}$ corrects the original one. So this fault is uncorrectable at the end, and the switch is not fault tolerant. This explains why simple switches fail. A time-independent evolution $e^{-i\lambda H}$, with $H$ a collective Hamiltonian commuting with the transversal gates, moves the encoder along a curve in $\M(n)$ that has no reason to stay in $\Dloc(n)$. The main result of this section is that staying in $\Dloc(n)$ is also sufficient.

\begin{theorem}
\label{thm:existence}
Let $\lambda \mapsto [A_\lambda]$, $\lambda \in [0,1]$, be a continuous path in $\Dloc(n)$. There is an error-transparent switch $U$ with $U(\lambda) A_0 = A_\lambda$, where the encoders $A_\lambda$ represent the points of the path and depend continuously on $\lambda$. If the path is piecewise smooth, $U$ may be chosen piecewise smooth.
\end{theorem}

The idea of the proof is as follows. By Theorem~\ref{thm:gram}, the error states $E_\mu A_\lambda e_i$ have the same inner products for every $\lambda$, so there is a unitary taking the error states of $A_0$ to those of $A_\lambda$. We extend this unitary to the rest of $\mathcal{H}_n$ by a continuous rotation of the remaining directions, and then check that the result commutes with permutations and with the transversal gates.

\begin{proof}
Because the path is continuous and $[0,1]$ is simply connected, it lifts through the bundle of unit vectors over $\PP(\M(n))$, so we may choose encoders $A_\lambda$ representing its points that depend continuously on $\lambda$. Let $\kappa = 1 + 3n$ and let $T_\lambda$ be the linear map from $\C^\kappa \otimes W$ to $\mathcal{H}_n$ with $T_\lambda(f_\mu \otimes e_i) = E_\mu A_\lambda e_i$, where $f_\mu$ is the standard basis of $\C^\kappa$. Theorem~\ref{thm:gram} says that $T_\lambda^\dagger T_\lambda = \Gamma \otimes \id$ for every $\lambda$. Since $\Gamma$ is invertible, the map $Y_\lambda = T_\lambda (\Gamma^{-1/2} \otimes \id)$ preserves inner products, its range is the error space $\mathcal{F}_\lambda$ of $\mathcal{C}_{A_\lambda}$, and $Y_\lambda$ depends continuously on $\lambda$. The operator $Y_\lambda Y_0^\dagger$ maps $\mathcal{F}_0$ onto $\mathcal{F}_\lambda$ preserving inner products, and
\begin{equation}
  Y_\lambda Y_0^\dagger\, T_0 = Y_\lambda Y_0^\dagger Y_0 (\Gamma^{1/2}\otimes\id) = T_\lambda,
\end{equation}
which is~\eqref{eq:ET}.

To extend to the directions orthogonal to $\mathcal{F}_0$, let $P_\lambda$ be the orthogonal projection onto $\mathcal{F}_\lambda$. For two projections $P, P'$ of equal rank with $\lVert P - P' \rVert < 1$, the operator $\Omega = P'P + (\id - P')(\id - P)$ is invertible, and its unitary part $\Omega (\Omega^\dagger \Omega)^{-1/2}$ maps the range of $P$ onto the range of $P'$ and the orthogonal complement onto the orthogonal complement. Dividing $[0,1]$ into short intervals and composing these unitaries gives a continuous family of unitaries $K_\lambda$ with $K_0 = \id$ that carry the orthogonal complement of $\mathcal{F}_0$ onto that of $\mathcal{F}_\lambda$. This construction goes back to Kato~\cite{Kato1950,KatoBook}; for a differentiable path one may instead integrate the generator $i[\partial_\lambda P_\lambda, P_\lambda]$ on the complement, which avoids the subdivision. We set
\begin{equation}
  U(\lambda) = Y_\lambda Y_0^\dagger + K_\lambda (\id - P_0).
\end{equation}
This is unitary, continuous, equal to $\id$ at $\lambda = 0$, and satisfies~\eqref{eq:ET}. If the path is piecewise smooth, so are $Y_\lambda$, $P_\lambda$, and $K_\lambda$.

It remains to check the symmetries. Let $g$ be a qubit permutation or one of $D(\hat X), D(\hat Z), D(\hat F)$. Applying $g$ to an error state gives a combination of error states with coefficients independent of $\lambda$. For a permutation $\pi$ this holds because $\pi \sigma^a_l \pi^\dagger = \sigma^a_{\pi(l)}$ and $\pi A_\lambda = A_\lambda$, and for a transversal gate $D(k)$ because $D(k) \sigma^a_l D(k)^\dagger = \sum_b R(k)_{ba} \sigma^b_l$ and $D(k) A_\lambda = A_\lambda k$. There is therefore a fixed unitary $\tau(g)$ on $\C^\kappa \otimes W$ with $g T_\lambda = T_\lambda \tau(g)$ for all $\lambda$. Since $\tau(g)^\dagger (\Gamma \otimes \id) \tau(g) = T_\lambda^\dagger g^\dagger g T_\lambda = \Gamma \otimes \id$, the matrix $\Gamma \otimes \id$ commutes with $\tau(g)$, and so does $\Gamma^{-1/2} \otimes \id$. Hence $g Y_\lambda = Y_\lambda \tau(g)$, and $g$ commutes with $Y_\lambda Y_0^\dagger$ and with $P_\lambda$. The operators $\Omega$ in the construction of $K_\lambda$ are built from the projections $P_\lambda$, so they commute with $g$, and so do their unitary parts. Therefore $g$ commutes with $U(\lambda)$.
\end{proof}

Together, Propositions~\ref{prop:ET} and~\ref{prop:necessary} and Theorem~\ref{thm:existence} give a complete answer, in which pointwise and operational fault tolerance turn out to have the same existence criterion.

\begin{corollary}
\label{cor:classification}
Let $[A_0], [A_1] \in \Dloc(n)$. The following are equivalent.
\begin{enumerate}
  \item There is a single-fault-tolerant switch from $\mathcal{C}_{A_0}$ to $\mathcal{C}_{A_1}$, so that one recovery corrects a single-qubit fault at an unknown time.
  \item There is a pointwise single-fault-tolerant switch from $\mathcal{C}_{A_0}$ to $\mathcal{C}_{A_1}$.
  \item There is an error-transparent switch from $\mathcal{C}_{A_0}$ to $\mathcal{C}_{A_1}$.
  \item The points $[A_0]$ and $[A_1]$ can be joined by a continuous path inside $\Dloc(n)$.
\end{enumerate}
\end{corollary}

\begin{proof}
Statement 1 implies statement 2 because each set $U_{\lambda \to 1} \Err_1 U_{\lambda \to 1}^\dagger$ is contained in $\Err_{\mathrm{sw}}$. Statement 2 implies statement 4 by Proposition~\ref{prop:necessary}, since $\lambda \mapsto [A_\lambda]$ is then a continuous path in $\Dloc(n)$. Statement 4 implies statement 3 by Theorem~\ref{thm:existence}, and statement 3 implies statement 1 by Proposition~\ref{prop:ET}.
\end{proof}

The question of fault-tolerant switching is thereby reduced to a question about the shape of the set $\Dloc(n)$, namely which of its points can be connected. In the next section, we answer this question for encoders with real coefficients. The equivalence holds within the class of Definition~\ref{def:switch}, whose unitaries commute with $D(\hat X)$, $D(\hat Z)$, and $D(\hat F)$ at every instant. Preserving the logical state only requires this at the endpoints, so a more general switch could pass through encoders outside $\M(n)$. Proposition~\ref{prop:necessary} uses no symmetry at all, so the trajectory of any pointwise single-fault-tolerant switch must consist of encoders of distance three, and Theorem~\ref{thm:existence} still applies to any path along which the error Gram matrix is constant. What is lost is Theorem~\ref{thm:gram}. If only $D(\hat X)$ and $D(\hat Z)$ are imposed, the computation of Proposition~\ref{prop:reduction} gives $A^\dagger J_a^2 A = d_a \id$ with three values $d_a$ constrained only by $d_x + d_y + d_z = j(j+1)$, since $\hat F$ no longer equates them, and the Gram matrix depends on these values. A sufficient condition for an error-transparent switch is then connectivity within a level set of $(d_x, d_y, d_z)$ in the larger space of encoders compatible with $\hat X$ and $\hat Z$; whether single-fault tolerance forces the Gram matrix to be constant in that setting is the same open question raised for higher distance in Sec.~\ref{sec:discussion}. With permutation symmetry alone, the Gram matrix depends on all of $A^\dagger \{J_a, J_b\} A$, and we have not pursued this case. We stress that Theorem~\ref{thm:existence} proves that the switching unitary exists and constructs it as an operator on $\mathcal{H}_n$. It does not decompose $U(\lambda)$ into physical control pulses or bound the resources such a decomposition requires; we return to this point in Sec.~\ref{sec:discussion}.

\section{Real encoders and connectivity}
\label{sec:real}

Write $\ket{0_L} = \sum_w a_w \ket{D_w}$ and $\ket{1_L} = \sum_w a'_w \ket{D_w}$ in the Dicke basis. An encoder is real if all of the coordinates $a_w$ and $a'_w$ are real numbers. We write $\bar A$ for the encoder with complex-conjugated coefficients.

\begin{proposition}
\label{prop:real}
If $A \in \M(n)$, then $\bar A \in \M(n)$. The same holds for the Clifford-code spaces $\M_{\rho_{1/2}}(n)$ and $\M_{\rho_{5/2}}(n)$ and the icosahedral-code spaces $\M_{\rho_{1/2}}(n)$ and $\M_{\rho'}(n)$. Consequently, each of these spaces has a basis of real encoders, and its real encoders form a real vector space of the same dimension.
\end{proposition}

\begin{proof}
Complex conjugation sends each generator to a product of generators. A direct computation gives
\begin{equation}
\begin{gathered}
  \overline{\hat X} = -\hat X, \quad \overline{\hat Z} = -\hat Z, \quad \overline{\hat F} = -\hat F \hat Z, \\
  \overline{\hat S} = \hat S^{-1}, \quad \overline{\Phi} = -\hat Z \Phi^{-1} \hat Z,
\end{gathered}
\end{equation}
and the last identity also holds with $\Phi'$ in place of $\Phi$. Since $n$ is odd, $D(-g) = -D(g)$. Conjugating $D(\hat X) A = A \hat X$ gives $-D(\hat X)\bar A = -\bar A \hat X$, so $D(\hat X) \bar A = \bar A \hat X$, and the same for $\hat Z$. Conjugating $D(\hat F) A = A \hat F$ gives $D(\hat F) D(\hat Z) \bar A = \bar A \hat F \hat Z$, and since $D(\hat Z) \bar A = \bar A \hat Z$ this becomes $D(\hat F) \bar A = \bar A \hat F$. Thus $\bar A \in \M(n)$. The relations for $\hat S$ and $\Phi$ pass to $\bar A$ in the same way. Every encoder $A$ in one of these spaces is then the combination $A = \mathrm{Re}\, A + i\, \mathrm{Im}\, A$ of the real encoders $\mathrm{Re}\,A = (A + \bar A)/2$ and $\mathrm{Im}\,A = (A - \bar A)/(2i)$ in the same space.
\end{proof}

For real encoders one of the two conditions of Proposition~\ref{prop:reduction} holds automatically.

\begin{theorem}
\label{thm:quadric}
For every real encoder $A \in \M(n)$ we have $c_2(A) = 0$, and $c_1$ restricts to a real quadratic form $q_n$ on the real encoders. The real distance-three codes are therefore the zero set
\begin{equation}
  \mathcal{Q}_n = \{[A] \mid A \in \M(n) \text{ real},\ q_n(A) = 0\} \subset \Dloc(n),
\end{equation}
where a real encoder determines a point of the real projective space of real encoders, and we identify $\mathcal{Q}_n$ with its image in $\Dloc(n)$ under the inclusion of real encoders into $\M(n)$.
Suppose $q_n$ is nondegenerate with $p$ positive and $m - p$ negative eigenvalues, where $1 \le p \le m - 1$. Then $\mathcal{Q}_n$ is connected if and only if $m \ge 3$.
\end{theorem}

\begin{proof}
In the computational basis, $J_x$ has real entries and $J_y$ has imaginary entries, so $\{J_x, J_y\}$ has purely imaginary entries. For real $\ket{0_L}$, the number $c_2(A) = \braket{0_L | \{J_x, J_y\} | 0_L}$ is then both real (by hermiticity) and imaginary, hence zero. The form $c_1(A) = \braket{0_L | J_z | 0_L}$ has a real symmetric matrix on real encoders because $J_z$ is real and diagonal.

For connectivity, choose real coordinates $x = (u, v) \in \R^p \times \R^{m-p}$ in which $q_n = \lVert u \rVert^2 - \lVert v \rVert^2$, and normalize $\lVert x \rVert = 1$. The zero set is then $\lVert u \rVert^2 = \lVert v \rVert^2 = \frac12$, a product of a sphere in $\R^p$ and a sphere in $\R^{m-p}$, and $\mathcal{Q}_n$ is obtained by identifying $x$ with $-x$. A sphere in $\R^d$ is connected for $d \ge 2$, while for $d = 1$ it is the pair of points $\pm 1$. If $p \ge 2$ and $m - p \ge 2$, the product is connected. If exactly one of them equals $1$, say $p = 1$, the zero set consists of two copies of a connected sphere, $u = \pm\frac{1}{\sqrt2}$, and the identification of $x$ with $-x$ maps one copy onto the other, so $\mathcal{Q}_n$ is connected. If $p = m - p = 1$, the zero set has four points and $\mathcal{Q}_n$ has two.
\end{proof}

For $m = 3$ and $q_n$ with two positive eigenvalues and one negative, the zero set is $u_1^2 + u_2^2 = v^2$, and after identifying $x$ with $-x$ it is a single circle, which is the situation on thirteen qubits. The full set $\Dloc(n)$ of complex encoders is larger than $\mathcal{Q}_n$, but the real encoders suffice for switching, as the next two corollaries show.

\begin{corollary}
\label{cor:switch}
Suppose $m \ge 3$ and $q_n$ is nondegenerate with eigenvalues of both signs. Any two real distance-three codes in $\M(n)$, and in particular any real Clifford code and any real icosahedral code of distance three on $n$ qubits, are joined by an error-transparent switch.
\end{corollary}

\begin{proof}
Both codes lie in $\mathcal{Q}_n$, which is connected by Theorem~\ref{thm:quadric}, and Corollary~\ref{cor:classification} provides the switch.
\end{proof}

\begin{corollary}
\label{cor:complex}
Let $\rho$ be one of the two Clifford or two icosahedral choices, and let $\mathcal{Z}_\rho(n)$ be the set of distance-three codes in $\M_\rho(n)$. If $\mathcal{Z}_\rho(n)$ is nonempty, it is connected and contains a real code. Consequently, under the hypotheses of Corollary~\ref{cor:switch}, any two distance-three Clifford or icosahedral codes, real or complex, are joined by an error-transparent switch.
\end{corollary}

\begin{proof}
By Corollary~\ref{cor:larger}, $\mathcal{Z}_\rho(n)$ is the set where $c_1 = 0$, and it is everything for $\rho'$. By Proposition~\ref{prop:real}, we may use a basis of real encoders, in which $c_1(x) = x^\dagger C x$ for a real symmetric matrix $C$ and complex coordinates $x$. After a real change of coordinates, $c_1(x) = \lVert x_+ \rVert^2 - \lVert x_- \rVert^2$, where $x = (x_+, x_-, x_0)$ is split according to the positive, negative, and zero eigenvalues of $C$. If $x_+$ or $x_-$ is absent, the normalized zero set is the unit sphere in the $x_0$ coordinates, which is connected. Otherwise, it consists of the vectors with $\lVert x_+ \rVert^2 = \lVert x_- \rVert^2 = t$ and $\lVert x_0 \rVert^2 = 1 - 2t$ for $t \in [0, \frac12]$. For each $t$ these form a product of connected spheres, and moving $t$ continuously to $\frac12$ joins every point to the connected set at $t = \frac12$. In both cases, the zero set is connected, and it contains real vectors whenever it is nonempty because $C$ is real. Each complex code in $\mathcal{Z}_\rho(n) \subset \Dloc(n)$ is therefore joined within $\Dloc(n)$ to a real code, and Corollary~\ref{cor:switch} joins the real codes.
\end{proof}

A path inside $\mathcal{Z}_\rho(n)$ consists of codes compatible with all Clifford gates, or with all icosahedral gates. The proof of Theorem~\ref{thm:existence} applies without change to such a path, so along it, the switch may be chosen to commute with all of these transversal gates. When $m = 2$, the situation is rigid, and the complex distance-three codes coincide with the real ones.

\begin{proposition}
\label{prop:rigid}
Suppose $m = 2$ and that neither $c_2$ nor $q_n$ vanishes identically. Then $\Dloc(n) = \mathcal{Q}_n$, and $\Dloc(n)$ consists of at most two points.
\end{proposition}

\begin{proof}
Choose a basis $B_1, B_2$ of real encoders and write an encoder as $A = x_1 B_1 + x_2 B_2$ with $x \in \C^2$. Then $c_2(A) = x^\dagger H_2 x$, where $(H_2)_{kl} = \braket{B_k e_0 | \{J_x, J_y\} | B_l e_0}$ is purely imaginary because $B_k e_0$ is real and $\{J_x, J_y\}$ has purely imaginary entries. A Hermitian $2 \times 2$ matrix with purely imaginary entries is a real multiple of $\sigma^y$, so $H_2 = h\sigma^y$ with $h \ne 0$, and $c_2(A) = 2h\, \mathrm{Im}(\bar x_1 x_2)$. Hence $c_2(A) = 0$ exactly when $\bar x_1 x_2$ is real, which for nonzero $x$ means that $x$ is a complex multiple of a real vector. Every distance-three code is therefore real, and the real ones are the zeros of the nonzero quadratic form $q_n$ in two real variables, of which there are at most two up to scaling.
\end{proof}

The two-dimensional cases of the survey are $n = 7$, $11$, and $15$. In each of them, both forms are nonzero. For $n = 11$, the form $q_n$ has two negative eigenvalues, so $\Dloc(11)$ is empty, and for $n = 7$ and $n = 15$, it has one of each sign, so $\Dloc(n)$ consists of two points. By Corollary~\ref{cor:classification}, no single-fault-tolerant switch in the class of Definition~\ref{def:switch} connects distinct distance-three codes in these cases. For $n = 7$, there is no Clifford code of distance three, as Table~\ref{tab:survey} shows, so the question of switching between Clifford and icosahedral codes does not arise. The two points of $\Dloc(7)$ are the icosahedral code of Pollatsek and Ruskai~\cite{PollatsekRuskai2004} and the encoder $D(\hat S) A \hat S^\dagger$ obtained from it, which is compatible with the conjugated group $\hat S\, \BI\, \hat S^\dagger$, and the result says that these two codes are not joined by any fault-tolerant switch that commutes with the transversal gates throughout. Imposing only $D(\hat X)$ and $D(\hat Z)$ along the path does not help. We computed the set of distance-three codes in the four-dimensional space of seven-qubit encoders compatible with $\hat X$ and $\hat Z$ alone, and it again consists of exactly these two isolated points.

\begin{table*}[t]
\caption{Survey for odd $n \le 29$. The column $m$ gives the dimension of $\M(n)$, and $\mathrm{sig}(q_n)$ gives the numbers of positive and negative eigenvalues of the quadratic form $q_n$ of Theorem~\ref{thm:quadric}. For each of the four subspaces $\M_\rho(n)$ we list its dimension and, in parentheses, the numbers of positive and negative eigenvalues of $q_n$ restricted to it. A real distance-three code exists in $\M_\rho(n)$ exactly when these are not all of one sign, which includes the case of $\M_{\rho'}(n)$, where $q_n$ vanishes by Corollary~\ref{cor:larger}. In every row the two entries of $\mathrm{sig}(q_n)$ sum to $m$, so $q_n$ is nondegenerate, as Theorem~\ref{thm:quadric} requires. The two Clifford columns sum to $\mathrm{sig}(q_n)$ in every row, because $D(\hat S)$ commutes with $J_z$ and so the subspaces $\M_{\rho_{1/2}}(n)$ and $\M_{\rho_{5/2}}(n)$, which together span $\M(n)$, are orthogonal for $c_1$. The last column indicates whether a real Clifford code and a real icosahedral code of distance three exist and are joined by an error-transparent switch.}
\label{tab:survey}
\centering
\begin{tabular}{rrcccccc}
\toprule
$n$ & $m$ & $\mathrm{sig}(q_n)$ & $\BO,\ \rho_{1/2}$ & $\BO,\ \rho_{5/2}$ & $\BI,\ \rho_{1/2}$ & $\BI,\ \rho'$ & switch \\
\midrule
5  & 1 & (0,1) & 0          & 1\ (0,1)  & 0         & 0 & no \\
7  & 2 & (1,1) & 1\ (0,1)   & 1\ (1,0)  & 0         & 1 & no \\
9  & 1 & (1,0) & 1\ (1,0)   & 0         & 0         & 0 & no \\
11 & 2 & (0,2) & 1\ (0,1)   & 1\ (0,1)  & 1\ (0,1)  & 0 & no \\
13 & 3 & (2,1) & 1\ (1,0)   & 2\ (1,1)  & 1\ (1,0)  & 1 & yes \\
15 & 2 & (1,1) & 1\ (0,1)   & 1\ (1,0)  & 0         & 0 & no \\
17 & 3 & (1,2) & 2\ (1,1)   & 1\ (0,1)  & 0         & 1 & yes \\
19 & 4 & (2,2) & 2\ (1,1)   & 2\ (1,1)  & 1\ (0,1)  & 1 & yes \\
21 & 3 & (2,1) & 1\ (1,0)   & 2\ (1,1)  & 1\ (1,0)  & 0 & no \\
23 & 4 & (1,3) & 2\ (0,2)   & 2\ (1,1)  & 1\ (0,1)  & 1 & yes \\
25 & 5 & (3,2) & 3\ (2,1)   & 2\ (1,1)  & 1\ (1,0)  & 1 & yes \\
27 & 4 & (2,2) & 2\ (1,1)   & 2\ (1,1)  & 0         & 1 & yes \\
29 & 5 & (2,3) & 2\ (1,1)   & 3\ (1,2)  & 1\ (0,1)  & 1 & yes \\
\bottomrule
\end{tabular}
\end{table*}

Table~\ref{tab:survey} summarizes the survey. The smallest qubit number admitting an error-transparent switch between a Clifford code and an icosahedral code is $n = 13$, and such switches exist for every odd $n$ between $13$ and $29$ except $n = 15$ and $n = 21$, where no icosahedral code of distance three exists. At $n = 15$, there is no icosahedral encoder at all. At $n = 7$, the icosahedral code exists but is isolated.

\section{Thirteen qubits}
\label{sec:thirteen}

We now carry out the construction explicitly for $n = 13$, the smallest case in Table~\ref{tab:survey}.

\paragraph{The space of encoders.} It is convenient to remove the phases from~\eqref{eq:Mdef}. For a gate $g$, with $\sqrt{\det g}$ the principal root, we have $D(\hat g) = (\sqrt{\det g})^{-13} D(g)$ and $A \hat g = (\sqrt{\det g})^{-1} A g$, so $D(\hat g) A = A \hat g$ is equivalent to $D(g) A = (\det g)^{6} A g$. With $\det X = \det Z = -1$, and $\det F = i$, the space $\M(13)$ consists of the maps with
\begin{equation}
  D(X) A = A X, \qquad D(Z) A = A Z, \qquad D(F) A = -A F .
  \label{eq:frame13}
\end{equation}
Write $\ket{0_L} = A e_0 = \sum_w a_w \ket{D_w}$. The first relation says $\ket{1_L} = D(X)\ket{0_L}$, which replaces $w$ by $13 - w$. The second relation says $D(Z)\ket{0_L} = \ket{0_L}$, and since $D(Z)\ket{D_w} = (-1)^w \ket{D_w}$, it forces $a_w = 0$ for odd $w$.

For the third relation we need $D(F) = D(H) D(S^\dagger)$ on Dicke states. The phase part is simple, $D(S^\dagger)\ket{D_w} = (-i)^w \ket{D_w}$. For the Hadamard part, $H^{\otimes 13}\ket{x} = 2^{-13/2} \sum_y (-1)^{x \cdot y} \ket{y}$, where $x \cdot y$ counts the positions in which both strings have a one. Fix $y$ with $v$ ones and sum the signs over all $x$ with $w$ ones. If $i$ of the ones of $x$ fall on ones of $y$, there are $\binom{v}{i}\binom{13-v}{w-i}$ such strings, each with sign $(-1)^i$, so the sum is the Krawtchouk polynomial~\cite{macwilliams1977theory}
\begin{equation}
  K_w(v) = \sum_i (-1)^i \binom{v}{i}\binom{13-v}{w-i}.
\end{equation}
Collecting the strings $y$ of weight $v$ into Dicke states gives
\begin{equation}
  H^{\otimes 13} \ket{D_w} = 2^{-13/2} \sum_v K_w(v) \sqrt{\binom{13}{v} \Big/ \binom{13}{w}}\, \ket{D_v}.
\end{equation}
Set $b_w = a_w / \binom{13}{w}^{1/2}$. Using $F e_0 = (e_0 + e_1)/\sqrt2$ and comparing the coefficients of $\ket{D_v}$ for even $v$ on both sides of $D(F)\ket{0_L} = -(\ket{0_L} + \ket{1_L})/\sqrt2$, the third relation becomes the integer eigenvalue problem
\begin{equation}
  \sum_{w \text{ even}} (-1)^{w/2} K_w(v)\, b_w = -64\, b_v,
\end{equation}
for $v \in \{0, 2, \ldots, 12\}$. The coefficients with odd $v$ give the same equations again, because $K_w(13 - v) = (-1)^w K_w(v) = K_w(v)$ for even $w$. The solution space is three-dimensional, in agreement with Table~\ref{tab:survey}, and in the coordinates $(b_0, b_2, \ldots, b_{12})$ it is spanned by
\begin{equation}
\begin{aligned}
  p &= (-165, 0, 6, 0, 5, 0, 0), \\
  q &= (10, 0, -1, 0, 0, 0, 5), \\
  r &= (0, 11, 0, -1, 0, 3, 0).
\end{aligned}
\end{equation}
The real encoders are the combinations $b = \alpha p + \beta q + \gamma r$ with real $\alpha, \beta, \gamma$.

\paragraph{The distance-three circle.} Since $J_z \ket{D_w} = (\frac{13}{2} - w)\ket{D_w}$, we have $c_1 = \sum_w a_w^2 (\frac{13}{2} - w)$, and $c_2$ is computed from the matrix elements of $\{J_x, J_y\}$ between Dicke states. For complex $\alpha, \beta, \gamma$ we find
\begin{align}
  c_1 &\propto 1485\,\lvert\alpha\rvert^2 - 330\,\mathrm{Re}(\bar\alpha\beta) + 5\,\lvert\beta\rvert^2 + 264\,\lvert\gamma\rvert^2, \label{eq:conic}\\
  c_2 &\propto \mathrm{Im}\bigl[\bar\gamma\,(3\alpha - \beta)\bigr],
\end{align}
with positive constants. For real coefficients, $c_2$ vanishes, as Theorem~\ref{thm:quadric} requires, and $c_1 = 0$ becomes $1485\alpha^2 - 330\alpha\beta + 5\beta^2 + 264\gamma^2 = 0$. Completing the square gives $5(\beta - 33\alpha)^2 + 264\gamma^2 = 3960\alpha^2$, an ellipse in the plane $\alpha = 1$, which becomes the unit circle in the rescaled coordinates $(\beta - 33\alpha)/(6\sqrt{22}\alpha)$ and $\gamma/(\sqrt{15}\alpha)$ of Fig.~\ref{fig:conic}. We refer to it as the circle $\mathcal{Q}_{13}$. The real distance-three codes on thirteen qubits are given by the points
\begin{equation}
  (\alpha, \beta, \gamma) = \bigl(1,\ 33 + 6\sqrt{22}\cos\theta,\ \sqrt{15}\sin\theta\bigr),
  \label{eq:param}
\end{equation}
with $\theta \in [0, 2\pi)$. Every point of this circle is a code with transversal $X$, $Z$, and $F$, and its distance is exactly three.

\begin{proposition}
\label{prop:exact3}
Every code in $\mathcal{Q}_{13}$ has distance exactly three.
\end{proposition}

\begin{proof}
The distance is at least three because the codes lie in $\Dloc(13)$. For distance four, every weight-three error $P$ would have to satisfy $\braket{0_L | P | 0_L} = \braket{1_L | P | 1_L}$. For the Pauli strings $Z_1 Z_2 Z_3$ and $X_1 X_2 Z_3$ we have $D(X) P D(X) = -P$, and $\ket{1_L} = D(X)\ket{0_L}$, so this would require $\braket{0_L | P | 0_L} = 0$. Computing the Dicke matrix elements, these two conditions read $f = 0$ and $g = 0$, up to nonzero factors, with
\begin{align}
  f &= 6975\alpha^2 - 750\alpha\beta - 25\beta^2 + 504\gamma^2, \\
  g &= 225\alpha^2 - 450\alpha\beta + 25\beta^2 + 408\gamma^2 - (480\alpha + 160\beta)\gamma .
\end{align}
On the circle $\alpha \ne 0$, since $\alpha = 0$ would force $\beta = \gamma = 0$ in~\eqref{eq:conic}, so we set $\alpha = 1$. Eliminating $\gamma^2$ between $f = 0$ and~\eqref{eq:conic} gives $19\beta^2 + 66\beta - 2277 = 0$, with roots $\beta = (-33 \pm 24\sqrt{77})/19$. The negative root gives $\gamma^2 < 0$ in~\eqref{eq:conic}, so $f$ vanishes on the circle only at the two points with $\beta_* = (-33 + 24\sqrt{77})/19$ and $\gamma = \pm\gamma_*$, where $\gamma_*^2 = (600\sqrt{77} - 3675)/361 > 0$. At these points $225 - 450\beta_* + 25\beta_*^2 + 408\gamma_*^2 = 0$ exactly, so $g = -(480 + 160\beta_*)\gamma$, which is nonzero. Hence at every point of the circle one of the two weight-three errors violates the condition.
\end{proof}

\begin{figure}[t]
\centering
\includegraphics[width=\columnwidth]{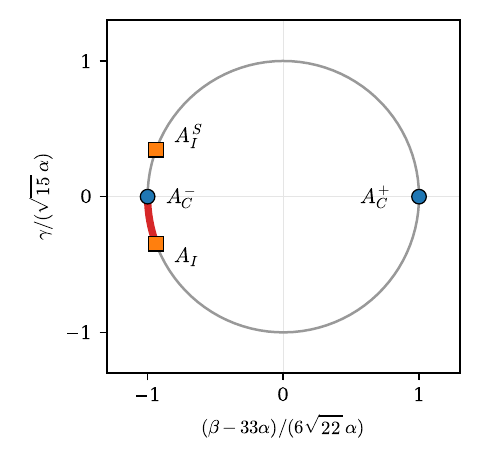}
\caption{The circle $\mathcal{Q}_{13}$ of real distance-three encoders, drawn in the plane $\alpha = 1$. The two Clifford codes $A_C^\pm$ sit at $\theta = 0$ and $\theta = \pi$, the icosahedral code $A_I$ at $\theta_I$, and $A_I^S = D(\hat S) A_I \hat S^\dagger$ at $-\theta_I$. The highlighted arc from $A_C^-$ to $A_I$ is the path along which we construct the error-transparent switch in Sec.~\ref{sec:numerics}.}
\label{fig:conic}
\end{figure}

\paragraph{Clifford and icosahedral codes on the circle.} Four distinguished points lie on the circle; see Fig.~\ref{fig:conic}. Since $D(S)\ket{D_w} = i^w \ket{D_w}$, the relation $D(S) A = A S$ holds exactly when $\ket{0_L}$ uses only weights $w$ divisible by $4$, that is, when $\gamma = 0$. The points $\theta = 0$ and $\theta = \pi$ are therefore Clifford codes, with
\begin{equation}
  \beta/\alpha = 33 \pm 6\sqrt{22},
\end{equation}
and in terms of the normalized gates, they satisfy $D(\hat S) A = -A\hat S$, the choice $\rho_{5/2}$ of Sec.~\ref{sec:setting}. We denote them by $A_C^\pm$. In Dicke coordinates, $A_C^-$ has $(b_0, b_4, b_8, b_{12}) = (165 - 60\sqrt{22},\ -27 + 6\sqrt{22},\ 5,\ 165 - 30\sqrt{22})$. It coincides, up to a global sign, with the $((13,2,3))$ code of Ref.~\cite{Kubischta2025Intrinsic}, whose logical states are given there in the same Dicke labeling. The code $A_C^+$ is different, with $\lvert \braket{A_C^+, A_C^-} \rvert = \frac{1}{14}$, so the thirteen-qubit Clifford code of distance three is not unique. Allowing complex coefficients, the Clifford codes form a whole circle in the two-dimensional space $\M_{\rho_{5/2}}(13)$, of which $A_C^\pm$ are the two real points. By Corollary~\ref{cor:complex} and the remark after it, this circle joins $A_C^+$ to $A_C^-$ through Clifford codes, and along it the error-transparent switch commutes with all transversal Clifford gates.

The icosahedral code sits at $\cos\theta_I = -\sqrt{22}/5$ and $\sin\theta_I = -\sqrt{3}/5$, where $(\alpha, \beta, \gamma) \propto (5, 33, -3\sqrt{5})$. Its Dicke coordinates are
\begin{equation}
\begin{split}
  (b_0, b_2, \ldots, b_{12}) = (&-495,\ -33\sqrt5,\ -3,\ 3\sqrt5, \\ &25,\ -9\sqrt5,\ 165),
\end{split}
\end{equation}
and on it $D(\Phi)$ implements the logical gate $\Phi'$, as Corollary~\ref{cor:larger} predicts. The space $\M_{\rho'}(13)$ is one-dimensional, so $A_I$ is the only icosahedral code of this kind. The reflection $\gamma \mapsto -\gamma$ of the circle is the map $A \mapsto D(\hat S) A \hat S^\dagger$, which preserves $\M(13)$ because conjugation by $\hat S$ maps $\BT$ to itself. It sends $A_I$ to $A_I^S = D(\hat S) A_I \hat S^\dagger$ at $-\theta_I$, an icosahedral code for the conjugated group $\hat S\, \BI\, \hat S^\dagger$.

The arc of the circle from $A_C^-$ at $\theta = \pi$ to $A_I$ at $\theta_I = \pi + \arcsin(\sqrt3/5)$ spans about $20.3^\circ$ in the parametrization~\eqref{eq:param}. The overlaps $\lvert \braket{A_C^\pm, A_I} \rvert$ are $0.264$ for $A_C^+$ and $0.808$ for $A_C^-$, so $A_C^-$ is the natural partner of the icosahedral code. Corollary~\ref{cor:switch} gives error-transparent switches from both Clifford codes to $A_I$.

\section{Numerical construction and fault analysis}
\label{sec:numerics}

We constructed the switch of Theorem~\ref{thm:existence} along the arc from $A_C^-$ to $A_I$ and compared it with two simpler switches whose trajectories leave $\Dloc(13)$. The comparison illustrates Proposition~\ref{prop:necessary} rather than benchmarking the switches against one another, since the two comparison switches are generated by low-rank collective Hamiltonians while the error-transparent switch is not, as shown at the end of this section. All three commute with qubit permutations and with $D(\hat X), D(\hat Z), D(\hat F)$, and all three map $A_C^-$ to $A_I$ to within $1.3 \times 10^{-12}$ or better.

The first comparison switch is the rotation in $\M(13)$ that turns $A_C^-$ into $A_I$ within the plane they span, through the angle $\arccos 0.808 \approx 36.1^\circ$, and acts as the identity on all states orthogonal to $A_C^- e_i$ and $A_I e_i$. We call it the geodesic switch. The second is a sequence of four pulses generated by two collective Hamiltonians that commute with the transversal gates, the symmetrized cubic $K_3 = \frac16\sum_{\pi \in S_3} J_{\pi(x)} J_{\pi(y)} J_{\pi(z)}$ and the quartic $Q_4 = J_x^4 + J_y^4 + J_z^4$. The product $e^{-it_4 Q_4} e^{-it_3 K_3} e^{-it_2 Q_4} e^{-it_1 K_3}$ with $(t_1, t_2, t_3, t_4) = (0.16871072, -0.16703516, 0.12142938, 0.06190127)$ maps $A_C^-$ to $A_I$ with infidelity $1.3 \times 10^{-12}$, where a negative time means evolution under $-Q_4$ for the duration $\lvert t_i \rvert$, whereas the best product of three or fewer alternating pulses we found leaves an infidelity of $0.02$. This solution is sensitive to the pulse times. To second order, the switching infidelity is a quadratic form in the timing errors whose eigenvalues range from $2.7 \times 10^2$ to $1.5 \times 10^5$, and relative timing errors of $10^{-4}$ and $10^{-3}$ produce mean switching infidelities of $3.4 \times 10^{-5}$ and $3.4 \times 10^{-3}$ over random perturbations. Such timing errors are control errors of the kind analyzed in Proposition~\ref{prop:control}. A timing error $\delta$ in pulse $i$ displaces the final encoder by $-i\delta\, U_{>i} H_i A_{\le i}$, where $H_i \in \{K_3, Q_4\}$, $A_{\le i}$ is the encoder after pulse $i$, and $U_{>i}$ is the evolution that follows. The component of this displacement along the final encoder is a first-order global phase of the logical state, and the component orthogonal to it is the displacement $B$ of Proposition~\ref{prop:control}. Proposition~\ref{prop:control} then gives a logical Pauli error with probability $\ell_i \delta^2$ out of a total displacement $d_i \delta^2$, with $(\ell_i) = (3.5 \times 10^2,\ 1.0 \times 10^5,\ 5.5 \times 10^2,\ 7.0 \times 10^3)$ and $(d_i) = (3.9 \times 10^2,\ 1.1 \times 10^5,\ 1.3 \times 10^3,\ 1.5 \times 10^5)$. For the first two pulses, about $90\%$ of the displacement becomes a logical Pauli error, for the third about $41\%$, and for the last about $5\%$.

The computations do not require the full $2^{13}$-dimensional space. It suffices to work in the $38$-dimensional subspace spanned by $V_{13}$ and by the states $\sigma^a_k \ket{D_w}$ with $k \in \{1, 2\}$. This subspace is mapped to itself by all collective operators, and by permutation symmetry, faults on qubits $1$ and $2$ determine the behavior of faults on every qubit; Appendix~\ref{app:numerics} gives the details. We divided the arc into $300$ steps and built $U(\lambda)$ from the maps $Y_\lambda$ of Theorem~\ref{thm:existence} and the transport of the complement, using the ten operators $\id$, $J_x, J_y, J_z$, $\sigma^a_1$, and $\sigma^a_2$ as the error frame. Along the path, the inner products of the frame states agree with their values at the start to $5 \times 10^{-14}$, $U(\lambda)$ is unitary to $3 \times 10^{-14}$, and in Frobenius norm, it commutes with $D(X)$, $D(Z)$, $D(F)$, and the exchange of qubits $1$ and $2$ to $2 \times 10^{-13}$ and with the total spin $J_x^2 + J_y^2 + J_z^2$ to $1 \times 10^{-12}$. At the endpoint $A_I^\dagger U(1) A_C^- = \id$ to the same precision. The generator $H_\theta = i\, \partial_\theta U\, U^\dagger$ of this switch is a collective operator, that is, a polynomial in $J_x, J_y, J_z$, and only odd degrees occur, as they must for a real path, because $U(\theta)$ is then real in the computational basis, so $H_\theta$ is imaginary and odd under time reversal. Its parts of degree $3$, $7$, $9$, $11$, and $13$ carry squared weights $33\%$, $41\%$, $18\%$, $0.6\%$, and $7\%$ at the start of the arc and $46\%$, $22\%$, $26\%$, $0.3\%$, and $6\%$ at the end. Its integrated action $\int \lVert H_\theta \rVert\, d\theta$ is $0.94$ in units where the $J_a$ are dimensionless, against $4.4 \times 10^2$ for the pulse sequence, whose bare times understate its cost because $\lVert K_3 \rVert = 43$ and $\lVert Q_4 \rVert = 1.9 \times 10^3$ on the relevant sectors. The complement transport is a free choice, so these numbers describe one generator among many. To remove this freedom, we solved the error-transparency condition directly for the generator. Differentiating~\eqref{eq:ET} along the path gives $-i H_\theta\, E A_\theta = E\, \partial_\theta A_\theta$ for all $E \in \Err_1$, a linear condition on $H_\theta$, which we solved by least squares at several points of the arc over the collective $\BT$-invariant Hermitian operators of degree at most $d$. No generator of degree at most $11$ is error transparent on this arc, the smallest relative residual being $0.17$, while degree $13$ gives residual $10^{-10}$. A generator of degree $3$ cannot even keep the encoder on $\Dloc(13)$, and truncating the generator above to degree $3$ gives an endpoint infidelity of $0.135$. A low-order collective Hamiltonian therefore cannot reproduce an error-transparent switch along this arc.

For each switch, we applied a Pauli fault $\sigma^a_1$, with $a \in \{x,y,z\}$, at a sequence of times, carried the resulting state through the rest of the switch, and applied the recovery of the icosahedral code described in Appendix~\ref{app:numerics}. That recovery measures which syndrome subspace the state is in, and it treats a state outside the error space of the code as a heralded erasure. We report the entanglement infidelity of the recovered logical channel with erasure counted as failure, averaged over the Pauli type and over fault times sampled uniformly in the switch parameter, together with the leakage probability, which is the probability of erasure. Infidelities conditioned on no erasure are given in Appendix~\ref{app:numerics}. The error-transparent switch gives a mean infidelity of $1.0 \times 10^{-15}$ and a worst case of $2.9 \times 10^{-14}$, as Proposition~\ref{prop:ET} predicts. The geodesic switch gives a mean of $0.127$ and a worst case of $0.352$, with mean leakage $0.118$, so most of its failures are heralded. The pulse sequence gives a mean of $0.850$ and a worst case of $0.992$, with mean leakage $0.478$. Figure~\ref{fig:faults} shows how the infidelity depends on the fault time. The geodesic switch is exact only for a fault after it has finished. A fault at its start hits the source code, which has distance three, so the fault is correctable at that moment; the remaining rotation, however, does not carry the error space of $A_C^-$ onto that of $A_I$, and the error arrives outside the error space of the target code. At intermediate times both comparison switches also pass through encoders without distance three, where Proposition~\ref{prop:necessary} applies. The pulse sequence leaves $\Dloc(13)$ immediately and returns only at the endpoint.

\begin{figure}[t]
\centering
\includegraphics[width=\columnwidth]{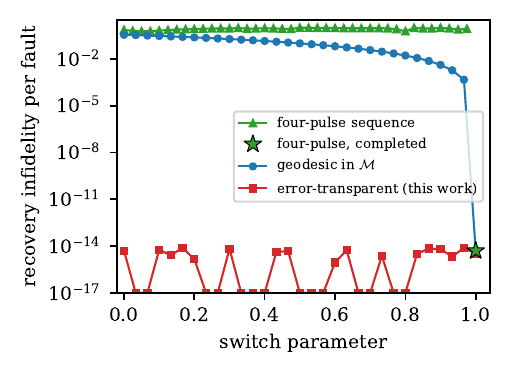}
\caption{Entanglement infidelity after recovery of the icosahedral code, with states outside its error space counted as erasure, for a single Pauli fault on one qubit at a given fraction of the switch, averaged over the three Pauli types. The horizontal axis is the path parameter for the error-transparent and geodesic switches, and for the pulse sequence it is $(i - 1 + f)/4$ for a fault at fraction $f$ of pulse $i$. The geodesic switch and the four-pulse sequence both map $A_C^-$ to $A_I$ to numerical precision but leave the set of distance-three codes on the way. The error-transparent switch of Theorem~\ref{thm:existence} is recovered to machine precision at every fault time. The star marks a fault after the four-pulse switch has completed, which the target recovery corrects; it is not part of the forty-sample average reported in the text. Values below $10^{-17}$, which are zero to machine precision, are drawn at that floor.}
\label{fig:faults}
\end{figure}

\section{Discussion}
\label{sec:discussion}

\paragraph{What the switch requires physically.} The $n$-qubit space splits into sectors of fixed total spin $s$, and an operator that commutes with all qubit permutations acts within each sector, in the same way on each of its copies~\cite{GoodmanWallach}. A code state lies in the sector $s = j$, and a single-qubit error adds a component in the sector $s = j - 1$, so error transparency constrains only these two sectors, and the switch may act as the identity on all others. It is therefore specified by its action $U_j$ on the first sector and $U_{j-1}$ on the second, and error transparency constrains both. Every such pair is realized by a polynomial in $J_x, J_y, J_z$, for example by using the projections onto the two sectors, which are polynomials in $J_x^2 + J_y^2 + J_z^2$. The action on $V_n$ does not determine $U_{j-1}$, so control that addresses only the symmetric subspace is not enough to implement the switch. Collective control of this kind is available through geometric phase gates and one-axis twisting~\cite{Ouyang2025Switch,OuyangBrennen2026}. Our results establish that the switching unitary exists and construct it as an operator, not as a control sequence. Synthesizing it from available collective interactions with bounded resources is the natural next step, and we have not addressed it here.

\paragraph{Control errors.} The fault model covers single-qubit faults on the data at arbitrary times during the switch. It does not cover errors in the control fields themselves. Control errors that commute with the transversal gates, such as a wrong amplitude or duration of a pulse of $K_3$ or $Q_4$, never act on the logical state, by the argument after Definition~\ref{def:switch}. They only move the encoder to a nearby point of $\M(n)$, possibly one without distance three. Their effect on recovery is controlled by one vector in $\M(n)$. Let $c_1(A, B) = \braket{A e_0 | J_z | B e_0}$, so that $c_1(A) = c_1(A, A)$, and set
\begin{equation}
  N(A) = \sum_{a} J_a A \sigma^a .
\end{equation}

\begin{proposition}
\label{prop:control}
Let $[A] \in \Dloc(n)$ with $\braket{A, A} = 1$. Then $N(A) \in \M(n)$, $\braket{A, N(A)} = 0$, $\braket{N(A), N(A)} = j(j+1)$, and $c_1(A, B) = \frac13 \braket{N(A), B}$ for all $B \in \M(n)$. Let $B \in \M(n)$ satisfy $\braket{A, B} = 0$ and $\braket{B, B} = 1$, let $A_\varepsilon = \cos\varepsilon\, A + \sin\varepsilon\, B$, and let $\hat N = N(A)/\sqrt{j(j+1)}$. For every logical state $\psi$, the standard recovery of $\mathcal{C}_A$ applied to $A_\varepsilon \psi$ returns $\psi$ with probability $\cos^2\varepsilon$, returns $\sigma^a \psi$ with probability $\frac13 \sin^2\varepsilon\, \lvert\braket{\hat N, B}\rvert^2$ for each $a \in \{x, y, z\}$, and with the remaining probability $\sin^2\varepsilon\,(1 - \lvert\braket{\hat N, B}\rvert^2)$ finds a state outside the error space of $\mathcal{C}_A$.
\end{proposition}

\begin{proof}
For $k \in \BT$ we have $D(k) J_a = \sum_b R(k)_{ba} J_b D(k)$ and $D(k) A = A k$, so $D(k) N(A) = \sum_b J_b A \bigl(\sum_a R(k)_{ba}\, k \sigma^a k^\dagger\bigr) k = N(A)\, k$, using $\sum_a R(k)_{ba}\, k \sigma^a k^\dagger = \sigma^b$. Hence $N(A) \in \M(n)$. The computation in the proof of Proposition~\ref{prop:reduction}, applied to $A^\dagger J_a B$ in place of $A^\dagger J_a A$, gives $A^\dagger J_a B = c_1(A, B) \sigma^a$. Therefore $\braket{N(A), B} = \sum_a (\sigma^a A^\dagger J_a B)_{00} = 3 c_1(A, B)$, and in particular $\braket{A, N(A)} = 3 c_1(A) = 0$. On $\Dloc(n)$, Proposition~\ref{prop:reduction} gives $A^\dagger J_a J_b A = \frac{j(j+1)}{3}\delta_{ab}\id$, so $\braket{N(A), N(A)} = j(j+1)$.

For the recovery, write $B = \braket{\hat N, B} \hat N + B_\perp$. Then $\braket{A, B_\perp} = 0$ and $c_1(A, B_\perp) = 0$, so $A^\dagger \sigma^a_k B_\perp = \frac2n A^\dagger J_a B_\perp = 0$, and $B_\perp \psi$ is orthogonal to every error state of $\mathcal{C}_A$. On the other hand $\hat N \psi = \sum_a J_a A \sigma^a \psi / \sqrt{j(j+1)}$, and $J_a A \sigma^a \psi = \frac12 \sum_k \sigma^a_k A \sigma^a\psi$ is the symmetric sum of the errors $\sigma^a_k$ applied to the encoded state $\sigma^a \psi$. The three vectors $J_a A \sigma^a \psi$ are mutually orthogonal, orthogonal to $A\psi$, and have squared norm $\frac{j(j+1)}{3}$, so the recovery maps each of them to $\sigma^a \psi$ with the stated probabilities.
\end{proof}

A control error of size $\varepsilon$ therefore never rotates the logical state coherently. It causes a logical Pauli error with probability at most $\sin^2\varepsilon$, and none at all when the displacement $B$ is orthogonal to $N(A)$. In that case, the displaced component is not harmless; it lies outside the error space of $\mathcal{C}_A$ and is lost with probability $\sin^2\varepsilon$, which a recovery that flags states outside the error space reports as a heralded erasure. The vector $N(A)$ plays the role of a normal vector to the set of distance-three codes, since $c_1(A, B) = \frac13\braket{N(A), B}$, and for a real encoder on $\mathcal{Q}_n$, it is real and perpendicular to the circle or quadric. Tangency alone does not decide which displacements cause logical errors. The direction $i\hat N$ satisfies $\mathrm{Re}\, c_1(A, i\hat N) = 0$, so it is tangent to the hypersurface $\{c_1 = 0\}$, yet it produces logical Pauli errors at the full rate $\sin^2\varepsilon$; the displacements that produce none are exactly those orthogonal to $\hat N$. For the record, $i\hat N$ is not tangent to the set $\Dloc(n)$ itself, since that also requires $\mathrm{Re}\, c_2(A, i\hat N) = 0$, and for a real encoder, $c_2(A, \hat N)$ is purely imaginary, so $c_2(A, i\hat N)$ is real and generally nonzero. At $A = A_I$ it equals $-12.5$. On thirteen qubits we verified both cases numerically for $A = A_I$. A displacement along the tangent of $\mathcal{Q}_{13}$ produces no logical Pauli error, and a displacement along $\hat N(A_I)$ produces one with probability $\sin^2\varepsilon$. Successive control errors of this kind add up within $\M(n)$, and each round of recovery on the target code returns the encoder to $A$ except for the probabilities above. A data fault that follows a control error of size $\varepsilon$ acts on an encoder that is off the intended trajectory by $O(\varepsilon)$, and the remaining evolution is off by $O(\varepsilon)$ as well. By error transparency of the ideal switch, the fault then reaches the target code as the ideal single-qubit error plus a remainder of size $O(\varepsilon)$, so the combined failure probability is $O(p\,\varepsilon^2)$ for a fault probability $p$, which is of second order in the joint error budget when $p$ and $\varepsilon^2$ are counted at the same order. Control errors that do not commute with the transversal gates require a separate analysis. A second limitation concerns the synthesis of the switch. Error transparency is a property of the unitary $U(\lambda)$ acting on data qubits subject to single-qubit faults. Once $U(\lambda)$ is decomposed into physical operations, the fault model has to be revisited, because an interaction of cubic or quartic degree in the collective spin realized through a gate decomposition can turn one physical fault into a correlated error on several qubits, which is outside $\Err_1$ and is not covered by Proposition~\ref{prop:ET}. This is the main risk for a switch of this kind. Error transparency as defined in~\eqref{eq:ET} concerns faults at the points $U(\lambda)$ of the path. If the switch is realized as a sequence of segments taking $U(\lambda_i)$ to $U(\lambda_{i+1})$, a single-qubit fault between segments is covered, while a fault partway through a segment acts at an operator that is not on the path and is not covered unless the segment itself is error transparent. A decomposition would therefore have to be error transparent segment by segment, and even then only against the single-qubit faults it produces. Implementations in which the collective interaction is mediated by a global coupling, such as a common bosonic mode, keep the data faults single-qubit but introduce faults of the mediator, whose propagation needs its own analysis, as in Ref.~\cite{Ouyang2025Switch}.

\paragraph{Beyond distance three and other groups.} Theorem~\ref{thm:gram} is special to distance three, for the reason given in Remark~\ref{rem:structural}. For a code of distance five, the relevant inner products involve polynomials of degree four in the $J_a$, and the tetrahedral group has invariants of degree three and four, namely $K_3$ and $Q_4$ above, whose values on the code are not fixed by symmetry and may therefore vary among distance-five codes. The construction of Theorem~\ref{thm:existence} still applies to any path of codes along which the error Gram matrix for the relevant errors is constant, since that is the only property of $\Dloc(n)$ the proof uses, and it then yields an error-transparent switch. Whether single-fault tolerance forces the Gram matrix to be constant, as it does at distance three through Proposition~\ref{prop:reduction}, is a separate question that we have not settled, so the equivalence of Corollary~\ref{cor:classification} is not claimed beyond distance three. The same caution applies at distance three for gate groups that fix an axis of the Bloch sphere, such as the binary dihedral groups that give transversal $T$ gates~\cite{Kubischta2025}, because there the value of $A^\dagger J_z^2 A$ is not fixed by symmetry. Remark~\ref{rem:structural} turns this into a criterion that can be checked for any finite group of gates from the invariant polynomials of its rotation image.

\paragraph{Place in a full protocol.} Error transparency addresses the switch, not the rest of the computation. Fault-tolerant error correction and state preparation for PI codes, which are not stabilizer codes, remain the main obstacles to a complete protocol, and recent recovery algorithms for arbitrary PI codes~\cite{OuyangBrennen2026} are the natural starting point. The switch needs no ancilla code blocks or measurements between codes, as teleportation-based switching would, and it removes the stabilizer code from the measurement-free scheme of Ref.~\cite{Ouyang2025Switch}. The Clifford group together with any non-Clifford gate generates a dense set of single-qubit gates~\cite{ConwaySmith}, so the Clifford code and the icosahedral code together provide a universal single-qubit gate set, generated by the transversal Clifford gates on one code and the transversal icosahedral gates on the other, with an error-transparent switch between them. Kubischta and Teixeira showed that universal gate sets built from the binary icosahedral group together with one additional gate approximate single-qubit unitaries with shorter sequences than Clifford$+T$~\cite{Kubischta2023}; a compilation analysis for the gate set considered here is left open.

\paragraph{Related mechanisms.} Our construction parallels error-transparent gates for bosonic and superconducting qubits~\cite{Vy2013,Kapit2018,Ma2020}, where the gate Hamiltonian commutes with the error operators on the code space, and fault-tolerant holonomic computation~\cite{Oreshkov2009}, where slow evolution within a code space implements logical gates. The same commutation holds here. For a differentiable switch with generator $H_\lambda = i\,\dot U(\lambda) U(\lambda)^\dagger$, differentiating the error-transparency condition~\eqref{eq:ET} gives $[H_\lambda, E] A_\lambda = 0$ for every $E \in \Err_1$, so the switching Hamiltonian commutes with every single-qubit error on the instantaneous code space. The difference lies in what the condition is used for. In those works, the code is fixed and the Hamiltonian implements a logical gate, whereas here the code itself moves, and Theorem~\ref{thm:gram} supplies the existence of a Hamiltonian satisfying the condition along the whole path. The same idea applies to code deformation more broadly~\cite{BombinMartinDelgado2009,Vuillot2019}, for continuous families of codes whose error Gram matrix is constant, provided the resulting maps between error spaces can be realized with the available physical controls.

\section{Conclusion}
\label{sec:conclusion}

We have shown that switching between permutation-invariant codes with transversal Clifford and icosahedral gates is governed by the shape of the set of distance-three codes in a shared space of encoders. The transversal gates $X$, $Z$, and $F$ reduce the distance-three conditions to two quadratic equations, every code satisfying them has the same error geometry, and a switch that tolerates a single-qubit fault at an unknown time exists if and only if the two codes are joined by a path of distance-three codes, in which case the switch can be chosen error transparent. For real codes, this set is the zero set of one quadratic form, and on thirteen qubits, it is an explicit circle joining two Clifford codes to a code with transversal icosahedral gates. The switch is constructed as an operator, and its decomposition into physical operations remains open.

The broader lesson concerns the relation between the geometry of a family of codes and the fault tolerance of moving within it. Fault-tolerant switching is possible here because the correctable-error geometry stays constant along the family, so that the error spaces of all codes on a path can be identified with one another and fault tolerance becomes a question of connectivity. Constant error geometry is a strong condition that generic families of codes do not satisfy. Remark~\ref{rem:structural} gives a criterion for when it holds, namely the absence of low-degree invariant polynomials for the rotation image of the gate group. Where it holds, for other finite groups, higher spin, or qudit systems, the same reasoning turns the classification of fault-tolerant deformations into a question about the shape of a set defined by polynomial equations, provided the maps between error spaces can be realized with the available controls. Finding such families, and extending the construction beyond distance three, are natural directions for future work.

\bibliographystyle{unsrt}
\bibliography{refs}

@article{EastinKnill2009, 
author={Eastin, Bryan and Knill, Emanuel}, 
title={Restrictions on Transversal Encoded Quantum Gate Sets}, journal={Phys. Rev. Lett.}, 
volume={102}, 
pages={110502}, 
year={2009}}

@article{ZengCrossChuang2011, 
author={Zeng, Bei and Cross, Andrew and Chuang, Isaac L.}, title={Transversality Versus Universality for Additive Quantum Codes}, 
journal={IEEE Trans. Inf. Theory}, 
volume={57}, 
pages={6272--6284}, 
year={2011}}

@article{BravyiKonig2013, 
author={Bravyi, Sergey and K{\"o}nig, Robert}, title={Classification of Topologically Protected Gates for Local Stabilizer Codes}, 
journal={Phys. Rev. Lett.}, 
volume={110}, 
pages={170503}, 
year={2013}}

@article{BravyiKitaev2005, 
author={Bravyi, Sergey and Kitaev, Alexei}, 
title={Universal Quantum Computation with Ideal {C}lifford Gates and Noisy Ancillas}, 
journal={Phys. Rev. A}, 
volume={71}, 
pages={022316}, 
year={2005}}

@article{Gidney2024, 
author={Gidney, Craig and Shutty, Noah and Jones, Cody}, title={Magic State Cultivation: Growing {T} States as Cheap as {CNOT} Gates}, 
journal={arXiv:2409.17595}, 
year={2024}}

@article{PaetznickReichardt2013, 
author={Paetznick, Adam and Reichardt, Ben W.}, 
title={Universal Fault-Tolerant Quantum Computation with Only Transversal Gates and Error Correction}, 
journal={Phys. Rev. Lett.}, 
volume={111}, 
pages={090505}, 
year={2013}}

@article{Anderson2014, 
author={Anderson, Jonas T. and Duclos-Cianci, Guillaume and Poulin, David}, 
title={Fault-Tolerant Conversion between the {S}teane and {R}eed-{M}uller Quantum Codes}, 
journal={Phys. Rev. Lett.}, 
volume={113}, 
pages={080501}, 
year={2014}}

@article{Bombin2015, 
author={Bomb{\'i}n, H{\'e}ctor}, 
title={Gauge Color Codes: Optimal Transversal Gates and Gauge Fixing in Topological Stabilizer Codes}, 
journal={New J. Phys.}, 
volume={17}, 
pages={083002}, 
year={2015}}

@article{Butt2024, 
author={Butt, Friederike and Heu{\ss}en, Sascha and Rispler, Manuel and M{\"u}ller, Markus}, 
title={Fault-Tolerant Code-Switching Protocols for Near-Term Quantum Processors}, 
journal={PRX Quantum}, 
volume={5}, 
pages={020345}, 
year={2024}}

@article{Heussen2025, 
author={Heu{\ss}en, Sascha and Hilder, Janine}, 
title={Efficient Fault-Tolerant Code Switching via One-Way Transversal {CNOT} Gates}, 
journal={Quantum}, 
volume={9}, 
pages={1846}, 
year={2025}}

@article{Pogorelov2024,
  title={Experimental fault-tolerant code switching},
  author={Pogorelov, Ivan and Butt, Friederike and Postler, Lukas and Marciniak, Christian D and Schindler, Philipp and M{\"u}ller, Markus and Monz, Thomas},
  journal={Nature Physics},
  volume={21},
  number={2},
  pages={298--303},
  year={2025},
  publisher={Nature Publishing Group UK London}
}

@article{Beverland2021, 
author={Beverland, Michael E. and Kubica, Aleksander and Svore, Krysta M.}, 
title={Cost of Universality: A Comparative Study of the Overhead of State Distillation and Code Switching with Color Codes}, journal={PRX Quantum}, 
volume={2}, 
pages={020341}, 
year={2021}}

@article{Ruskai2000, 
author={Ruskai, Mary Beth}, 
title={Pauli Exchange Errors in Quantum Computation}, journal={Phys. Rev. Lett.}, 
volume={85}, 
pages={194--197}, 
year={2000}}

@article{PollatsekRuskai2004, 
author={Pollatsek, Harriet and Ruskai, Mary Beth}, title={Permutationally Invariant Codes for Quantum Error Correction}, 
journal={Linear Algebra Appl.}, 
volume={392}, 
pages={255--288}, 
year={2004}}

@article{Ouyang2014, 
author={Ouyang, Yingkai}, 
title={Permutation-Invariant Quantum Codes}, 
journal={Phys. Rev. A}, 
volume={90}, 
pages={062317}, 
year={2014}}

@article{OuyangFitzsimons2016, 
author={Ouyang, Yingkai and Fitzsimons, Joseph}, 
title={Permutation-Invariant Codes Encoding More than One Qubit}, journal={Phys. Rev. A}, 
volume={93}, 
pages={042340}, 
year={2016}}

@article{Kubischta2023, 
author={Kubischta, Eric and Teixeira, Ian}, 
title={Family of Quantum Codes with Exotic Transversal Gates}, journal={Phys. Rev. Lett.}, 
volume={131}, pages={240601}, 
year={2023}}

@article{Kubischta2024, 
author={Kubischta, Eric and Teixeira, Ian}, 
title={Quantum Codes from Twisted Unitary $t$-Groups}, 
journal={Phys. Rev. Lett.}, 
volume={133}, 
pages={030602}, 
year={2024}}

@article{Kubischta2025, 
author={Kubischta, Eric and Teixeira, Ian}, 
title={Permutation-Invariant Quantum Codes with Transversal Generalized Phase Gates}, 
journal={IEEE Trans. Inf. Theory}, volume={71}, pages={485--498}, year={2025}}

@article{Kubischta2025Intrinsic, 
author={Kubischta, Eric and Teixeira, Ian}, 
title={Intrinsic Quantum Codes}, 
journal={arXiv:2511.14840}, 
year={2025}}

@article{Ouyang2025Switch, 
author={Ouyang, Yingkai and Jing, Yumang and Brennen, Gavin K.}, title={Measurement-Free Code-Switching Protocol for Low-Overhead Quantum Computation Using Permutation-Invariant Codes}, journal={PRX Quantum}, 
volume={6}, 
pages={040341}, 
year={2025}}

@article{Gross2021, 
author={Gross, Jonathan A.}, 
title={Designing Codes around Interactions: The Case of a Spin}, 
journal={Phys. Rev. Lett.}, 
volume={127}, 
pages={010504}, 
year={2021}}

@article{AydinAlbertBarg2026,
  title={Quantum error correction beyond SU (2): Spin, bosonic, and permutation-invariant codes from convex geometry},
  author={Aydin, Arda and Albert, Victor V and Barg, Alexander},
  journal={PRX Quantum},
  volume={7},
  number={1},
  pages={010341},
  year={2026},
  publisher={APS}
}

@article{OuyangBrennen2026, 
author={Ouyang, Yingkai and Brennen, Gavin K.}, 
title={A Theory of Quantum Error Correction for Permutation-Invariant Codes}, 
journal={arXiv:2602.13638}, 
year={2026}}

@article{KnillLaflamme1997, 
author={Knill, Emanuel and Laflamme, Raymond}, 
title={Theory of Quantum Error-Correcting Codes}, 
journal={Phys. Rev. A}, 
volume={55}, 
pages={900--911}, 
year={1997}}

@article{Vy2013, 
author={Vy, Os and Wang, Xiaoting and Jacobs, Kurt}, 
title={Error-Transparent Evolution: The Ability of Multi-Body Interactions to Bypass Decoherence}, 
journal={New J. Phys.}, 
volume={15}, 
pages={053002}, 
year={2013}}

@article{Kapit2018, 
author={Kapit, Eliot}, 
title={Error-Transparent Quantum Gates for Small Logical Qubit Architectures}, 
journal={Phys. Rev. Lett.}, 
volume={120}, 
pages={050503}, 
year={2018}}

@article{Ma2020,
  title={Error-transparent operations on a logical qubit protected by quantum error correction},
  author={Ma, Yuwei and Xu, Yuan and Mu, Xianghao and Cai, Weizhou and Hu, Ling and Wang, Weiting and Pan, Xiaoxuan and Wang, Haiyan and Song, YP and Zou, C-L and others},
  journal={Nature Physics},
  volume={16},
  number={8},
  pages={827--831},
  year={2020},
  publisher={Nature Publishing Group UK London}
}

@article{Kato1950, 
author={Kato, Tosio}, 
title={On the Adiabatic Theorem of Quantum Mechanics}, 
journal={J. Phys. Soc. Jpn.}, 
volume={5}, 
pages={435--439}, 
year={1950}}

@book{KatoBook, 
author={Kato, Tosio}, 
title={Perturbation Theory for Linear Operators}, 
publisher={Springer}, 
address={Berlin}, 
year={1995}}

@article{Oreshkov2009, 
author={Oreshkov, Ognyan and Brun, Todd A. and Lidar, Daniel A.}, 
title={Fault-Tolerant Holonomic Quantum Computation}, 
journal={Phys. Rev. Lett.}, 
volume={102}, 
pages={070502}, 
year={2009}}

@article{BombinMartinDelgado2009, 
author={Bomb{\'i}n, H{\'e}ctor and Martin-Delgado, Miguel A.}, 
title={Quantum Measurements and Gates by Code Deformation}, 
journal={J. Phys. A}, 
volume={42}, 
pages={095302}, 
year={2009}}

@article{Vuillot2019, 
author={Vuillot, Christophe and Lao, Lingling and Criger, Ben and Garc{\'i}a Almud{\'e}ver, Carmen and Bertels, Koen and Terhal, Barbara M.}, 
title={Code Deformation and Lattice Surgery Are Gauge Fixing}, 
journal={New J. Phys.}, 
volume={21}, 
pages={033028}, 
year={2019}}

@book{ConwaySmith, 
author={Conway, John H. and Smith, Derek A.}, 
title={On Quaternions and Octonions}, 
publisher={A K Peters}, 
address={Natick, MA}, 
year={2003}}

@book{GoodmanWallach, 
author={Goodman, Roe and Wallach, Nolan R.}, 
title={Symmetry, Representations, and Invariants}, 
publisher={Springer}, 
address={New York}, 
year={2009}}

@article{bradshaw2026introduction,
  title={Introduction to quantum error correction with stabilizer codes},
  author={Bradshaw, Zachary P and Dale, Jeffrey J and Evans, Ethan N},
  journal={Annals of Physics},
  volume={487},
  pages={170353},
  year={2026},
  publisher={Elsevier}
}

@article{bradshaw2026symmetry,
  title={Symmetry-based quantum codes beyond the {P}auli group},
  author={Bradshaw, Zachary P and LaBorde, Margarite L and Montero, Dillon},
  journal={Physical Review A},
  volume={113},
  number={4},
  pages={042431},
  year={2026},
  publisher={APS}
}

@article{bradshaw2026approximate,
  title={Approximate Error Correction for Quantum Simulations of {SU}(2) Lattice Gauge Theories},
  author={Bradshaw, Zachary P},
  journal={arXiv preprint arXiv:2603.26819},
  year={2026}
}

@book{macwilliams1977theory,
  title={The theory of error-correcting codes},
  author={MacWilliams, Florence Jessie and Sloane, Neil James Alexander},
  volume={16},
  year={1977},
  publisher={Elsevier}
}

\appendix

\section{Numerical methods}
\label{app:numerics}

Let $\mathcal{H}_{\mathrm{red}} \subset \mathcal{H}_{13}$ be the span of $V_{13}$ and of the states $\sigma^a_k \ket{D_w}$ for $k \in \{1, 2\}$, $a \in \{x, y, z\}$, and $0 \le w \le 13$. For fixed $k$, the states $\sigma^a_k \ket{D_w}$ together with $V_{13}$ span the space of states that are symmetric in the twelve qubits other than $k$, which has dimension $2 \cdot 13 = 26$. This space is the sum of $V_{13}$, of total spin $\frac{13}{2}$, and a $12$-dimensional space of total spin $\frac{11}{2}$ associated with qubit $k$. The total-spin-$\frac{11}{2}$ parts for $k = 1$ and $k = 2$ are two different copies of the same spin representation, labeled by vectors $f_1$ and $f_2$ that record which qubit carries the error, and these copies are linearly independent. Hence $\mathcal{H}_{\mathrm{red}}$ is the sum of $V_{13}$ and two copies of the spin-$\frac{11}{2}$ representation, of dimension exactly $14 + 2 \cdot 12 = 38$, in agreement with the numerically computed rank. A collective operator acts on $\mathcal{H}_{\mathrm{red}}$ by one matrix $U_{13/2}$ on $V_{13}$ and one matrix $U_{11/2}$ on both copies, so $\mathcal{H}_{\mathrm{red}}$ is mapped to itself, and we represent all collective operators exactly by their $38 \times 38$ restrictions.

The error-transparent switch was built on $\mathcal{H}_{\mathrm{red}}$ by dividing $\theta \in [\pi, \theta_I]$ into $300$ steps. At each step we formed $T_\theta$ from the ten operators $\id$, $J_x, J_y, J_z$, $\sigma^a_1$, and $\sigma^a_2$, computed $Y_\theta = T_\theta (G^{-1/2} \otimes \id)$ with $G$ the $10 \times 10$ Gram matrix, and advanced the transport of the complement by the unitary part of $P_{\theta'} P_\theta + (\id - P_{\theta'})(\id - P_\theta)$. The three collective operators are needed, and the reason also shows that the transport extends to a collective operator. Write a vector of $\mathcal{H}_{\mathrm{red}}$ as a component in $V_{13}$ plus a component in $V_{11/2} \otimes \mathrm{span}\{f_1, f_2\}$. The error states decompose as $\sigma^a_k A e_i = \frac{2}{13} J_a A e_i + t^a_i \otimes f_k$, with the same $t^a_i \in V_{11/2}$ for $k = 1$ and $k = 2$, so the span of $A e_i$, $J_a A e_i$, $\sigma^a_1 A e_i$, and $\sigma^a_2 A e_i$ is $\mathcal{V}_\theta \oplus (\mathcal{T}_\theta \otimes \mathrm{span}\{f_1, f_2\})$, where $\mathcal{V}_\theta \subset V_{13}$ is spanned by $A e_i$ and $J_a A e_i$ and $\mathcal{T}_\theta \subset V_{11/2}$ by the $t^a_i$. Its projector therefore has the form $P_\theta = P^{(13/2)}_\theta \oplus (P^{(11/2)}_\theta \otimes \id_2)$. Without the collective operators the $V_{13}$ components $\frac{2}{13} J_a A e_i$ cannot be separated from the spin-$\frac{11}{2}$ components, the span is not invariant under total spin, and the resulting transport does not commute with it. The partial isometry $Y_\theta Y_0^\dagger$ maps $J_a A_0 e_i \mapsto J_a A_\theta e_i$ within $V_{13}$ and $t^a_i(0) \otimes f_k \mapsto t^a_i(\theta) \otimes f_k$ with the same map for both $k$, so it has the form $V^{(13/2)} \oplus (V^{(11/2)} \otimes \id_2)$. Operators of this form make up an algebra closed under adjoints and under the functional calculus, so the operators $P_{\theta'} P_\theta + (\id - P_{\theta'})(\id - P_\theta)$, their unitary parts, and the final $U(\lambda)$ all have the form $U^{(13/2)} \oplus (U^{(11/2)} \otimes \id_2)$. This is what it means for the operator to act in the same way on both spin-$\frac{11}{2}$ copies, and it is exactly the structure of a collective operator restricted to $\mathcal{H}_{\mathrm{red}}$, so $U(\lambda)$ extends to a collective operator on $\mathcal{H}_{13}$, which we take to act as the identity on the total-spin sectors not reached by a single fault. Commutation with $J_x^2 + J_y^2 + J_z^2$ alone would not establish this, since it does not constrain the action on the multiplicity factor. In the construction reported here the two spin-$\frac{11}{2}$ blocks of $U(1)$ agree to $10^{-13}$, and the ten-operator frame has rank $20$ with Gram matrix constant along the path by Theorem~\ref{thm:gram} together with the identities $A^\dagger J_a J_b A = \frac{j(j+1)}{3}\delta_{ab}\id$ and $A^\dagger J_a \sigma^b_k A = \frac{2}{13} A^\dagger J_a J_b A$ of Proposition~\ref{prop:reduction} and the proof of Lemma~\ref{lem:PIKL}.

A fault $\sigma^a_1$ produces a state that is unchanged by permutations of qubits $2, \ldots, 13$. The recovery of $\mathcal{C}_{A_I}$ that we use measures the orthonormal syndrome subspaces of its error space and maps each back to the code, and it commutes with qubit permutations, so its action on such a state is determined by the part of the error space with the same symmetry. That part is spanned by $A_I e_i$, $\sigma^a_1 A_I e_i$, and $(2J_a - \sigma^a_1) A_I e_i$ and lies in $\mathcal{H}_{\mathrm{red}}$. We orthonormalize these vectors using the inverse square root of their Gram matrix, which has the form $g \otimes \id$ by Theorem~\ref{thm:gram}, to obtain syndrome vectors $s_\eta \otimes e_i$. For a faulty encoded map $F\colon W \to \mathcal{H}_{\mathrm{red}}$ we form the $2 \times 2$ matrices $\Delta_\eta = [\braket{s_\eta \otimes e_i | F e_{i'}}]_{ii'}$. The recovery is completed by sending the component of the state orthogonal to the error space to an orthogonal erasure flag, so that $\sum_\eta \Delta_\eta^\dagger \Delta_\eta = F^\dagger P_{\mathrm{err}} F$ need not equal $\id$. With this completion the entanglement fidelity of the recovered channel, for a maximally mixed logical input, is $\sum_\eta \lvert \mathrm{tr}\, \Delta_\eta \rvert^2 / 4$, the leakage probability is $p_{\mathrm{erase}} = 1 - \frac12 \mathrm{tr} \sum_\eta \Delta_\eta^\dagger \Delta_\eta$, and the fidelity conditioned on no erasure is $F_{e \mid \mathrm{survival}} = F_e / (1 - p_{\mathrm{erase}})$. The arithmetic mean of $1 - F_{e \mid \mathrm{survival}}$ over the sampled faults is $0.012$ for the geodesic switch and $0.718$ for the pulse sequence, while conditioning the pooled experiment on no erasure, which gives $1 - \bar F_e / (1 - \bar p_{\mathrm{erase}})$ in terms of the sample means, yields $0.010$ and $0.712$. The two averages weight the samples differently and neither combines with the mean leakage into the mean total infidelity. For a fault after the switch the fidelity equals $1$ to machine precision, which checks the procedure.

Fault times were sampled uniformly in the path parameter, at $31$ values of $\theta$ for the error-transparent switch and $31$ values of the rotation angle for the geodesic switch. For the pulse sequence we sampled each pulse at the fractions $f = 0, 0.1, \ldots, 0.9$ of its duration $\lvert t_i \rvert$ and placed the sample at $(i - 1 + f)/4$ on the horizontal axis of Fig.~\ref{fig:faults}; the reported means are averages over these $40$ samples and the three Pauli types. The pulse times were found by minimizing the switching infidelity on $\M(13)$ and then verified on $\mathcal{H}_{\mathrm{red}}$.

\end{document}